\documentclass[fontsize=12pt,paper=a4,leqno,emulatestandardclasses]{scrartcl}

\usepackage[left=2.5cm, right=3.5cm,
            top=2.5cm, bottom=1.5cm,
            includeheadfoot]{geometry}

\usepackage[T1]{fontenc}
\usepackage[utf8]{inputenc}

\usepackage{lmodern}
\usepackage[english]{babel}

\usepackage{amsmath, amsthm, amssymb, bm, dsfont}

\usepackage{paralist, enumerate, enumitem}

\usepackage{abstract}

\usepackage{natbib, hyperref}

\usepackage{dlfltxbcodetips}

\usepackage{bookmark}

\newcommand{\productspace}{\bigtimes_{i=1}^{n}\cX_{i}}
\newcommand{\dualproductspace}{\bigtimes_{i=1}^{n}\cX_{i}^{\prime}}

\newcommand{\Lzero}{L^{0}}
\newcommand{\Lzeron}{\nk{\Lzero}^{n}}
\newcommand{\Lone}{L^{1}}

\newcommand{\Linf}{L^{\infty}}

\newcommand{\Lp}{L^{p}}

\newcommand{\Mphi}{M^{\phi}}
\newcommand{\MPhi}{M^{\Phi}}
\newcommand{\Lphi}{L^{\phi}}
\newcommand{\LPhi}{L^{\Phi}}
\newcommand{\Mexp}{M^{\exp}}

\newcommand{\rhosing}{\rho_0}
\newcommand{\brisk}{\mathit{R}}
\newcommand{\oneton}{1,\dotsc,n}

\newcommand{\xixi}{\nk{\xi,\Xi}}
\newcommand{\xixiopt}{\nk{\xi^{\bar{X}},\Xi^{\bar{X}}}}

\newcommand{\xiopt}{\xi^{\bar{X}}}
\newcommand{\Xiopt}{\Xi^{\bar{X}}}
\newcommand{\Xioptj}{\Xi^{\bar{X},j}}

\newcommand{\bQopt}{\bQ^{\bar{X}}}

\usepackage{interval}
\intervalconfig{soft open fences}

\newcommand{\gk}[1]{\left \{ #1\right \}} 
\newcommand{\ek}[1]{\left[ #1\right]} 
\newcommand{\nk}[1]{\left(#1\right)} 

\newcommand{\pairing}[2]{\langle{#1,#2}\rangle}
\newcommand{\npairing}[2]{\langle{#1,#2}\rangle_{n}}

\newcommand{\norm}[1]{\Vert{#1}\Vert}
\newcommand{\abs}[1]{\vert{#1}\vert}

\newcommand{\zeroone}{\interval{0}{1}}

\newcommand{\ind}[1]{\mathds{1}_{#1}}

\newcommand{\bE}{\mathbb{E}}

\newcommand{\bN}{\mathbb{N}}

\newcommand{\bP}{\mathbb{P}}
\newcommand{\bQ}{\mathbb{Q}}
\newcommand{\bR}{\mathbb{R}}

\newcommand{\cA}{\mathcal{A}}

\newcommand{\cC}{\mathcal{C}}
\newcommand{\cD}{\mathcal{D}}
\newcommand{\cE}{\mathcal{E}}
\newcommand{\cF}{\mathcal{F}}

\newcommand{\cO}{\mathcal{O}}

\newcommand{\cR}{\mathcal{R}}
\newcommand{\cS}{\mathcal{S}}

\newcommand{\cU}{\mathcal{U}}
\newcommand{\cV}{\mathcal{V}}
\newcommand{\cW}{\mathcal{W}}
\newcommand{\cX}{\mathcal{X}}

\newcommand{\probspace}{\nk{\Omega,\cF,\bP}}
\newcommand{\measureablespace}{\nk{\Omega,\cF}}

\newcommand{\RN}[1]{\uppercase\expandafter{\romannumeral#1}} 

\numberwithin{equation}{section} 
\newtheorem{theo}[equation]{Theorem}
\newtheorem{prop}[equation]{Proposition}
\newtheorem{lemma}[equation]{Lemma}
\newtheorem{defi}[equation]{Definition}
\newtheorem{cor}[equation]{Corollary}
\newtheorem{ex}[equation]{Example}
\newtheorem{rem}[equation]{Remark}

\newtheorem*{theo*}{Theorem}
\newtheorem*{prop*}{Proposition}
\newtheorem*{lemma*}{Lemma}
\newtheorem*{defi*}{Definition}
\newtheorem*{cor*}{Corollary}
\newtheorem*{ex*}{Example}
\newtheorem*{rem*}{Remark}
\newtheorem*{rems*}{Remarks}
\newtheorem*{nota*}{Notation}
\newtheorem*{ques*}{Questions}
\newtheorem*{ass*}{Assumption}

\newcommand{\name}{Scalar systemic risk measures and Aumann-Shapley allocations}
\newcommand{\me}{Florian Schindler\\
                 Institute of Mathematics, Justus-Liebig-Universität Gießen\\
                 35392 Gießen, Germany;\\
                 email: florian.schindler@math.uni-giessen.de}
\newcommand{\joint}{Ludger Overbeck\\
                    Institute of Mathematics, Justus-Liebig-Universität Gießen\\
                    35392 Gießen, Germany;\\
                    email: ludger.overbeck@math.uni-giessen.de}

\begin{document}
\pagenumbering{arabic}
\setcounter{page}{1}
\begin{centering}
 \huge{\name}
 \\\vspace{8mm}
 \large{\joint}
 \\\vspace{4mm}
 \large{\me}\\
\end{centering}
\vspace{12mm}


\begin{abstract}
We study two different contributions to the theory of (scalar) systemic risk measures. Namely the first aggregate or axiomatic approach and the first inject capital approach. For this purpose we establish a general framework, which is rich enough to embed both approaches. It turns out that in most relevant situations systemic risk measures of the first inject capital approach have a representation in the more general axiomatic approach. Moreover, we study capital allocation rules (CARs). In both situations there exist canonical ways to answer the capital allocation problem. Additionally, a capital allocation rule (CAR) in the spirit of Aumann-Shapley is introduced, which gives us the opportunity to compute systemic capital allocations regardless of the risk measurement approach. This CAR also serves as an instrument to compare both approaches and to identify commonalities.

\end{abstract}


\section{Introduction}\label{sec:introduction}
The financial crises of the past decades showed that measuring systemic risk in a suitable way is an urgency. If we consider a firm consisting of \(n\in\bN{}\) business units or a portfolio consisting of \(n\in\bN{}\) assets, allocating the overall risk to its constituent parts plays a crucial role. In~\cite{kalkbrener2005axiomatic},~\cite{tasche2007capital} and~\cite{tsanakas2009split} the authors presented some helpful tools for this task.
However, we are not able to use these tools to compute and allocate the risk of a financial system. First, the idea that the constituent parts subsidize each other, i.e.\@ to simply sum up all profits and losses, is not a realistic scenario for a financial system \(\bar{X}=\nk{X_{1},\ldots,X_{n}}\). More general aggregation rules (ARs) \(\Lambda{}\) which reflect the structure of the system appropriately need to be considered. The authors in~\cite{chen2013axiomatic}, later extended to general measurable spaces by~\cite{kromer2016systemic}, and~\cite{biagini2019unified} developed this idea with two different approaches. The difference between these two approaches is the order of aggregating and allocating. Starting with some axioms \-- partially motivated by economic or risk management reasoning and partially by mathematical structures \-- it is shown  in~\cite{chen2013axiomatic} and~\cite{kromer2016systemic} that the risk of a financial system should have the form 
\begin{equation}\label{eq:syskromer}
 \rho\nk{\bar{X}}=\inf\gk{m\in\bR\mid\Lambda\nk{\bar{X}}-m\in\cA_{0}}=\rhosing\circ\Lambda\nk{\bar{X}},
\end{equation}
implying in particular that the aggregation takes place \textit{before} the capital is injected to the system. In other words, a mapping \(\rho{}\) fulfills the given axioms, if an only if it is of the form \(\rhosing\circ\Lambda{}\), for a suitable single firm risk measure \(\rhosing{}\) and a suitable AR \(\Lambda{}\). The scalar \(\rho\nk{\bar{X}}\) can be interpreted \-- in analogy to single firm risk measures \-- as the minimal amount of capital which has to be injected to the aggregated system to make it acceptable, i.e.\@ saves the system of a collapse. As mentioned earlier, this approach produces meaningful results, like dual representation theorems, in a general framework.
In~\cite{biagini2019unified} an alternative procedure is presented. They argue, that aggregating \textit{after} injecting capital on the level of the single institutions could affect the overall systemic risk in a positive way. Starting with a mapping as in~\eqref{eq:syskromer}, these risk measures appear as
\begin{equation}\label{eq:sysbiaginidet}
 D\nk{\bar{X}}=\inf\gk{\sum_{i=1}^{n}m_{i}\mid\bar{m}\in\bR^{n},\Lambda\nk{\bar{X}-\bar{m}}\in\cA_{0}}.
\end{equation}
Again, the set \(\cA_{0}\) is the acceptance set of a single-firm risk measure, the mapping \(\Lambda\colon\bR^{n}\to\bR{}\) is an AR\@ and the overall systemic risk is described by a scalar which can be interpreted as a valuation of the injected capital. Here, the valuation is given by adding up the injected capital, but more general valuations \(\pi{}\) are possible. An extension of~\eqref{eq:sysbiaginidet} occurs, if we allow for random capital injections, i.e.
\begin{equation}\label{eq:sysbiaginiran}
 \brisk\nk{\bar{X}}=\inf\gk{\sum_{i=1}^{n}Y_{i}\mid\bar{Y}\in\cC,\Lambda\nk{\bar{X}-\bar{Y}}\in\cA_{0}},
\end{equation}
where \(\cC{}\) is some set consisting of vectors of \(\cF{}\)-measurable functions \(Y_{i}\colon\Omega\to\bR{}\) on \(\probspace{}\) (denoted by \(\Lzeron=\nk{L^{0}\probspace{}}^{n}\)). The set \(\cC{}\) contains additional restrictions for the (possibly random) allocation \(\bar{Y}\). In all studies of these systemic risk measures, we have
\begin{equation}\label{eq:randomallocationswithdeterministicfuturevalue}
 \cC\subseteq\cC\nk{\bR}:=\gk{\bar{Y}\in\Lzeron\mid\sum_{i=1}^{n}Y_{i}\in\bR}.
\end{equation}
The interpretation of this constraint is, that the overall capital which is needed to save the system is determined today, but the allocation depends on the occurring scenario. These systemic risk measures fulfill the most basic properties one demands of a systemic risk measure. Under some restrictive assumptions further studies are possible and dual representation results can be derive. At first sight, it seems that this approach yields a completely new, more flexible type of systemic risk measures which are only related to the type~\eqref{eq:syskromer} in some trivial cases (for example if \(\Lambda\nk{\bar{x}}=\sum_{i=1}^{n}x_{i}\)). However, we are able to show in Theorem~\ref{theo:biaginisysnewproperties} that in most relevant situations systemic risk measures of type~\ref{eq:sysbiaginiran} can be embedded to the axiomatic approach. This means that we are able to represent them via~\eqref{eq:syskromer} for some suitable \(\tilde\rho_{0}\) and \(\tilde\Lambda{}\). This identification is demonstrated in Example~\ref{ex:lambdastarandpenaltyexponentialcase} and Remark~\ref{rem:biaginiidentification}.

We continue our analysis with allocation rules. Both systemic risk measures suggest a natural way to allocate the risk to the participants of the system. In Example~\ref{ex:penaltyallocationentropic} we demonstrate how it is possible to deduce a random allocation for a systemic risk measure of type~\eqref{eq:syskromer}. Additionally, we present an allocation rule, which is in the spirit of the famous Aumann-Shapley allocation rule (~\cite{denault2001coherent}) for single-firm risk measures. One can apply this rule to both types of systemic risk measures. Moreover, for Gâteaux differentiable systemic risk measures we can guarantee the full allocation property. Additionally, this CAR yields an alternative approach to compare both approaches and to identify commonalities.

Besides~\cite{chen2013axiomatic,kromer2016systemic} and~\cite{biagini2020fairness} there are many other contributions to the theory of (scalar) systemic risk measures.~\cite{arduca2021dual} and~\cite{ararat2020dual} present different approaches for the dual representation of both types. Additionally,~\cite{ararat2020dual} captures dual representations for set-values systemic risk measures. We would like to mention, that we do not consider set-valued systemic risk measures in this paper and refer the interested reader to~\cite{feinstein2017measures}. A completely different approach appears for example in~\cite{rogers2013failure}. Initiated by the seminal work of~\cite{eisenberg2001systemic}, the financial system is represented by a stochastic network with specific structure. One then tries to analyze the weaknesses of the network and how it might fail. The object of interest is a so called clearing vector which settles all liabilities in the network within a simultaneously clearing mechanism. Existence results and efficient algorithms to compute such a clearing vector are studied. A more detailed review of the existing literature with a focus on scalar systemic risk measures is given in\cite{doldi2021real}.

The paper is structured as follows: In Section~\ref{sec:modelandnotation} we present a very general framework subsuming most approaches to systemic risk measures, which allows us to introduce a general definition of a systemic risk measure in Section~\ref{sec:systemicriskmeasures}. Subsections~\ref{subsec:firstaggregate} and~\ref{subsecfirstinjectcapital} collect the most important results known from~\cite{kromer2016systemic} and~\cite{biagini2020fairness} about the different approaches. Moreover, Theorem~\ref{theo:biaginisysnewproperties} shows that in most relevant cases and all examples the risk measure of type~\ref{eq:sysbiaginiran} can be identified with risk measures of type~\ref{eq:syskromer}. This identification is highlighted in Example~\ref{ex:lambdastarandpenaltyexponentialcase} and Remark~\ref{rem:biaginiidentification}. Finally, in Section~\ref{sec:carsforsys} CARs for both types of systemic risk measures are studied and the CAR in the spirit of Aumann-Shapley is presented in order to keep the so-called full allocation property of CARs.


\section{Model and Notation}\label{sec:modelandnotation}

Let us start by introducing a  general setup. It will be rich enough for the main part of the theoretical results for single-firm risk measures and systemic risk measures. In some situation we restrict the setting to less general spaces (for example Banach lattices) to present meaningful special results. However, it can be seen as a minimal setup to derive meaningful representation results. So let \(\measureablespace{}\) be a general measurable space. By \(\cX{}\) we denote a linear space of \(\cF{}\)-measurable functions \(X\colon\Omega\to\bR{}\). With \(\cX_{+}\) we denote the cone of all non-negative functions, i.e.
\[
  \cX_{+}:=\gk{X\in\cX\mid{X}\nk{\omega}\geq0\forall\omega\in\Omega}.
\]
Then the relation 
\[
  X\succcurlyeq{Y}\Leftrightarrow{X-Y}\in\cX_{+}
\]
defines a partial order on \(\cX{}\). 

Since we are interested in dual representations we work with dual systems. For this purpose the technical framework of~\cite{aliprantis2006infinite} Chapter 5.14 is suitable. A dual system \(\pairing{\cE}{\tilde{\cE}}\) consists of a pair of vector spaces \(\cE{}\) and \(\tilde{\cE}\) together with a bilinear mapping, called pairing, 
\begin{equation}\label{eq:dualpairing}
  \nk{E,\tilde{E}}\mapsto{}\pairing{E}{\tilde{E}}
\end{equation}
which sends elements from \(\cE\times{}\tilde{\cE{}}\) to \(\bR{}\) and satisfies the two properties:
\begin{enumerate}[label = (\roman*)] 
  \item \(\pairing{E}{\tilde{E}}=0\) for all \(\tilde{E}\in\tilde{\cE}\Rightarrow{}E=0\),
  \item \(\pairing{E}{\tilde{E}}=0\) 
        for all \(E\in\cE\Rightarrow{}\tilde{E}=0\).
\end{enumerate} 
With this pairing the space \(\cE{}\) can be seen as a space of linear functionals on \(\tilde{\cE}\) and vice versa. To derive dual representation results (for convex functionals) one needs additional structure on \(\cE{}\), i.e.\@ a topology. A natural choice for a topology on \(\cE{}\) is the subspace topology of the product topology on the space of all real mappings on \(\tilde{\cE}\) denoted by \(\bR^{\tilde{\cE}}\). It is called weak topology denoted by \(\sigma\nk{\cE,\tilde{\cE}}\). With this topology the space \(\cE{}\) becomes a (or more precisely \(\nk{\cE,\sigma\nk{\cE,\tilde{\cE}}}\) is a) locally convex Hausdorff space. This is due to the fact that the product topology on the space of all real mappings on \(\tilde{\cE}\) is a locally convex Hausdorff topology, and it imparts its properties to its subspace topologies. 
The notation \(\sigma\nk{\cE,\tilde{\cE}}\) pays tribute to the fact that for a dual system \(\pairing{\cE}{\tilde{\cE}}\) the topological dual, denoted by \(\cE^{\prime}\), of \(\nk{\cE,\sigma\nk{\cE,\tilde{\cE}}}\) is exactly \(\tilde{\cE}\) (see Theorem 5.93\cite{aliprantis2006infinite}). Besides \(\sigma\nk{\cE,\tilde{\cE}}\) there are other locally convex topologies \(\tau{}\) on \(\cE{}\) which give \(\tilde{\cE}\) as the dual. We call these topologies consistent with the dual pair \(\pairing{\cE}{\tilde{\cE}}\). All these topologies are also Hausdorff (Lemma 5.97\cite{aliprantis2006infinite}), have the same closed convex sets and therefore the same lower semicontinuous (lsc) convex functions (see Lemma 5.97, Theorem 5.98 and Corollary 5.99\cite{aliprantis2006infinite}). The weak topology \(\sigma\nk{\cE,\tilde{\cE}}\) is the smallest (weakest or coarsest) consistent topology. The largest (strongest or finest) one is the so-called Mackey topology denoted by \(\tau\nk{\cE,\tilde{\cE}}\). We have for every consistent topology \(\tau{}\) that \(\sigma\nk{\cE,\tilde{\cE}}\subseteq\tau\subseteq\tau\nk{\cE,\tilde{\cE}}\) (see Theorem 5.113\cite{aliprantis2006infinite}). If we interchange the roles of \(\cE{}\) and \(\tilde{\cE}\) we are able to define topologies on \(\tilde{\cE}\) which give \(\cE{}\) as the dual. By \(\sigma\nk{\tilde{\cE},\cE}\) we denote the weak topology on \(\tilde{\cE}\) and so on.

Furthermore, as in~\cite{ruszczynski2006optimization} we need the following technical assumption:
\begin{enumerate}[label = (C)]
  \item\label{assumptionc} If \(E^{\prime}\not\in\cE^{\prime}_{+}\), then there exists \(E\in\cE_{+}\) such that \(\pairing{E}{E^{\prime}}<0\).
\end{enumerate}
~\cite{ruszczynski2006optimization} pointed out that this is a very mild requirement which ensures that the cone \(\cE_{+}^{\prime}\) is dual to the cone \(\cE_{+}\), i.e. \(\cE_{+}^{\prime}=\gk{E^{\prime}\in\cE^{\prime}:\pairing{E}{E^{\prime}}\geq0\forall{E}\in\cE_{+}}\). A sufficient condition to guarantee~\ref{assumptionc} is that the space \(\cE{}\) contains all indicator functions \(\ind{A}, A\in\cF{}\). From now on we always consider that the space \(\cX{}\) is equipped with a locally convex Hausdorff topology and assume that~\ref{assumptionc} holds true. Together with its topological dual \(\cX^{\prime}\) and the natural pairing we obtain the dual system \(\pairing{\cX{}}{\cX^{\prime}}\). 

In a risk management framework natural choices for \(\cX{}\) include the spaces \(\Lp\probspace{}\), \(1\leq{p}\leq\infty{}\), and, as a generalization, Orlicz spaces which are presented in~\ref{appendix:orliczspaces}. However, note that all crucial theoretical arguments to guarantee the decomposition and dual representation result (~\ref{theo:sysriskcomposition} and~\ref{theo:sysdualconvex}) do not require that \(\cX{}\) is a normed space. It is only necessary that \(\cX{}\) is a locally convex Hausdorff space, which justifies our choice.

We consider a model with one time period. Over this time period we are interested in measuring the risk of a financial system. For this purpose \(X\in\cX{}\) describes the loss of a single firm. For a system consisting of a finite set of \(n\) firms, representing \(n\) nodes in a financial network, the vector \(\bar{X}=\nk{X_1,\ldots,X_n}\in\productspace{}\) describes the losses of these nodes, i.e. \(X_{i}\in\cX_{i}\) is the loss of firm \(i\), where  \(\cX_{i}\) shares the same properties as \(\cX{}\). By abuse of notation we write for \(\bar{X},\bar{Y}\in\productspace{}\) \(\bar{X}\succcurlyeq\bar{Y}\) if \(X_{i}\succcurlyeq{Y}_{i}\) for all \(i=\oneton{}\) which means that the partial order is induced by
\[
  \nk{\productspace}_{+}:=\gk{\bar{X}\in\productspace\mid{X_{i}}\nk{\omega}\geq0\ \forall\omega\in\Omega\text{ and }{i}\in\gk{\oneton}}=\bigtimes_{i=1}^{n}\nk{\cX_{i}}_{+}.
\]
The bilinear function
\begin{equation}\label{eq:ndualpairing}
  \npairing{\bar{X}}{\bar{X}^{\prime}}:=\sum_{i=1}^{n}\pairing{X_{i}}{X^{\prime}_{i}},
\end{equation}
where \(\pairing{\cdot}{\cdot}\) is the pairing of the dual system \(\pairing{\cX_{i}}{\cX_{i}^{\prime}}\) from~\eqref{eq:dualpairing}, induces the dual system 
\[
  \pairing{\productspace}{\nk{\productspace}^{\prime}}.
\]
Obviously, \(\nk{\productspace}^{\prime}=\dualproductspace{}\) and the Cartesian product of locally convex Hausdorff spaces equipped with the product topology is again a locally convex Hausdorff space. Moreover, we will use the notation \(1_{n}=\nk{1,\dotsc,1}\in\bR^{n}\).


\section{Systemic Risk Measures}\label{sec:systemicriskmeasures}

Let us start by reviewing the most important results for the systemic risk measures of types~\eqref{eq:syskromer} and~\eqref{eq:sysbiaginidet} (resp.~\eqref{eq:sysbiaginiran}). Both risk measures have in common that they use a general aggregation rule (AR) to aggregate the overall risk into a univariate risk factor. The following definition characterizes these mappings.
\begin{defi}\label{def:ar}
 A mapping \(\Lambda\colon\bR^{n}\to\bR{}\) is called convex aggregation rule (AR) if it satisfies \(\Lambda\nk{\productspace}=\cX{}\)\footnote{By \(\Lambda\nk{\productspace}\) we mean \(\gk{\Lambda\circ\bar{X}\colon\Omega\to\bR\mid\bar{X}\in\productspace}\)} (respectively \(\Lambda\nk{\productspace}=\cX_{+}\) for \(\cR=\bR_{+}\)) and the following properties:
 \begin{enumerate}[label= (A\arabic*)]
  \item\label{armonoton} Monotonicity: If \(\bar{x}\geq\bar{y}\), then \(\Lambda\nk{\bar{x}}\geq\Lambda\nk{\bar{y}}\) for any \(\bar{x},\bar{y}\in\bR^{n}\).
  \item\label{arconvex} Convexity: \(\Lambda\nk{\alpha\bar{x}+\nk{1-\alpha}\bar{y}}\leq\alpha\Lambda\nk{\bar{x}}+\nk{1-\alpha}\Lambda\nk{\bar{y}}\) for any \(\bar{x},\bar{y}\in\bR^{n}\) and any \(\alpha\in\zeroone{}\).
  \item\label{arrsurjectivity} \(\cR{}\)-Surjectivity: \(\Lambda\nk{\bR^{n}}=\cR{}\) for \(\cR=\bR{}\) or \(\cR=\bR_{+}\).
 \end{enumerate}
 A mapping \(\Lambda\colon\bR^{n}\to\bR{}\) is called positively homogeneous aggregation rule (AR) if it satisfies \(\Lambda\nk{\productspace}\subseteq\cX{}\) (respectively \(\Lambda\nk{\productspace}\subseteq\cX_{+}\) for \(\cR=\bR_{+}\)), the properties~\ref{armonoton},~\ref{arconvex} and
 \begin{enumerate}[label= (A\arabic*)]
  \setcounter{enumi}{3}
  \item\label{arposhom} Positive homogeneity: \(\Lambda\nk{\alpha\bar{x}}=\alpha\Lambda\nk{\bar{x}}\) for any \(\bar{x}\in\bR^{n}\) and any \(\alpha\in\bR_{+}\).
  \item\label{arnormal} Normalization: \(\Lambda\nk{1_{n}}=n\).
 \end{enumerate}
\end{defi}

The interpretation of the properties~\ref{armonoton},\ref{arconvex},\ref{arposhom} and~\ref{arnormal} is straight forward as explained in~\cite{kromer2016systemic}.~\ref{arrsurjectivity} is a technical requirement which is strongly connected to the corresponding property of systemic risk measures and single-firm risk measures. However,  the most relevant ARs satisfy this property and therefore it is required as a basic property in the definition.

Before we move on to some examples, we present some sufficient conditions for \(\Lambda{}\) to satisfy the requirement \(\Lambda\nk{\productspace}\subseteq\cX{}\), respectively \(\Lambda\nk{\productspace}=\cX{}\) if we have \(\cX_{i}=\cX{}\) for all \(i\in\gk{\oneton{}}\) and \(\cX{}\) is a Banach lattice. The Lemma makes use of the function \(f_{\Lambda}\colon\bR\to\bR{}\) defined by 

\begin{equation}\label{eq:flambda}
  f_{\Lambda}\nk{a}:=\Lambda\nk{a\cdot1_{n}}.
\end{equation}

\begin{lemma}\label{lemma:technical requirement on lambda}
  Let \(\cX{}\) be a Banach lattice with norm \(\norm{\cdot}\). If any of the assumptions below hold then \(\Lambda\nk{\cX^{n}}=\cX{}\).
  \begin{enumerate}[label = (\alph*)]
    \item \(\Lambda\colon\bR^{n}\to\bR{}\) satisfies the properties~\ref{armonoton}-\ref{arrsurjectivity} and the function \(f_{\Lambda}\) satisfies \(\norm{f_{\Lambda}\nk{X}}<\infty{}\) and \(\norm{f^{-1}_{\Lambda}\nk{X}}<\infty{}\) for any \(X\in\cX{}\).
    \item \(\Lambda\colon\bR^{n}\to\bR{}\) satisfies the property~\ref{armonoton}, the function \(f_{\Lambda}\) is bijective and satisfies \(\norm{f_{\Lambda}\nk{X}}<\infty{}\) and \(\norm{f^{-1}_{\Lambda}\nk{X}}<\infty{}\) for any \(X\in\cX{}\).
    \item \(\Lambda\colon\bR^{n}\to\bR{}\) satisfies the property~\ref{arrsurjectivity}, is strictly increasing, continuous and the function \(f_{\Lambda}\) satisfies \(\norm{f_{\Lambda}\nk{X}}<\infty{}\) and \(\norm{f^{-1}_{\Lambda}\nk{X}}<\infty{}\) for any \(X\in\cX{}\).
  \end{enumerate}
\end{lemma}

\begin{proof}
  The Proof of Lemma 2.3 in~\cite{kromer2016systemic} easily extends to Banach lattices.
\end{proof}

We continue by presenting some examples.

\begin{ex}\label{ex:arlossesonly}[\cite{kromer2016systemic}]
 Obviously, simple summation
 \begin{equation}\label{eq:arsum}
  \Lambda^{\text{sum}}\nk{\bar{x}}=\sum_{i=1}^{n}x_{i}
 \end{equation}
 is an AR\@. If a shift with some parameter \(c\in\bR{}\) is allowed, we obtain
 \begin{equation}\label{eq:arsumshift}
  \Lambda^{\text{sum,c}}\nk{\bar{x}}=\sum_{i=1}^{n}x_{i}-c.
 \end{equation}
 These choices are only reasonable if the participants of the system cross-subsidize each other. An easy way to avoid cross-subsidization is to add up the losses only, i.e.
 \begin{equation}\label{eq:arlossesonly}
  \Lambda^{\text{loss}}\nk{\bar{x}}=\sum_{i=1}^{n}x_{i}^{+}.
 \end{equation}
 A simple modification occurs if we only consider losses beyond a given threshold level \(b>0\). The corresponding AR is then given by
 \begin{equation}\label{eq:arlossesonlythreshold}
  \Lambda^{\text{loss,b}}\nk{\bar{x}}=\sum_{i=1}^{n} \nk{x_{i} -b}^{+},\qquad b\in\bR_{+}.
 \end{equation}
\end{ex}

\begin{ex}\label{ex:arcritical}[\cite{kromer2016systemic}]
 It is also possible for the AR to take the structure of the given financial system into account. Suppose for example that there are some nodes, classified in a set \(A\subset\gk{\oneton}\), which would be dangerous for the stability of the system in case of a default. Then the AR
 \begin{equation}\label{eq:arcritical}
  \Lambda^{crit}\nk{\bar{x}}=\exp\nk{\gamma\sum_{i\in A}x_{i}^{+}}-1+\sum_{i\in N\setminus A}x_{i}^{+},\qquad \gamma>0
 \end{equation}
 penalizes large losses of the critical nodes on an exponential scale. Since for small \(x\) we have \(e^{x}-1\approx{x}\), small losses of the critical nodes are treated the same way as the losses of the less relevant nodes.
\end{ex}

\begin{ex}\label{ex:arutility}
 An important example in the setting of~\cite{biagini2020fairness} is given by
 \begin{equation*}
  \Lambda^{\text{ut}}\nk{\bar{x}}=\sum_{i=1}^{n}l_{i}\nk{x_{i}},
 \end{equation*}
 where each \(l_{i}\colon\bR\to\bR{}\) is a loss function, i.e.\ an increasing, convex function with \(\lim_{x\to\infty}\frac{l_{i}\nk{x}}{x}=\infty{}\). If we choose the exponential loss functions \(l_{i}\nk{x}=\frac{1}{\alpha_{i}}\exp\nk{\alpha_{i}x}\), \(\alpha_{i}>0\), the corresponding mapping \(\Lambda^{\text{exut}}\) is indeed a convex AR\@. 
\end{ex}

\begin{ex}\label{ex:arcontaigon}[\cite{chen2013axiomatic, kromer2016systemic}]
 The next AR is motivated by the structural contagion model of~\cite{eisenberg2001systemic} and was first introduced in~\cite{chen2013axiomatic}. In this setup the liabilities of firm i to firm j are captured within the so-called relative liability matrix \(\Pi=\nk{\Pi_{ij}}_{i,j=\oneton}\), where \(\Pi_{ij}\) represents the proportion of the total liabilities of firm i received by firm j. Moreover, we assume that an external regulator has the possibility to inject capital in the system. If the vector \(\bar{x}\in\bR^{n}\) captures the realized losses of all participants of the system, then each of them has two possibilities to cover these losses: One way is to receive money from the regulator and the other way is to reduce the payments to other firms by \(y_{i}\). If firm \(i\) decides to reduce their payments to firm \(j\), then firm \(j\) faces new losses \(\Pi_{ij}y_{i}\).
 \begin{equation}\label{eq:arcontaigon}
  \Lambda^{CM}(\bar{x})=\min_{\substack{b_{i}+y_{i}\geq{x}_{i}+\sum_{j=1}^n\Pi_{ji}y_{j}\\\forall i=\oneton,\bar{b},\bar{y}\in\bR^n_+}}\gk{\sum_{i=1}^n\nk{y_{i}+\gamma{b}_{i}}},\qquad\gamma>1.
 \end{equation}
\end{ex}

Let us move on by defining desirable properties for a systemic risk measure.

\begin{defi}\label{def:sysriskproperties}
 Consider the following properties for a mapping \(\rho\colon\productspace\to\bR\cup\gk{\infty{}}\):
 \begin{enumerate}[label= (S\arabic*)]
  \item\label{sysmonoton} Monotonicity: If \(\bar{X}\succcurlyeq\bar{Y}\), then \(\rho\nk{\bar{X}}\geq\rho\nk{\bar{Y}}\) for any \(\bar{X},\bar{Y}\in\productspace{}\).
  \item\label{sysconvex} Convexity:
        \begin{enumerate}[label = (S2\alph*)]
         \item\label{sysconvexoutcome} Outcome convexity: \(\rho\nk{\alpha\bar{X}+\nk{1-\alpha}\bar{Y}}\leq\alpha\rho\nk{\bar{X}}+\nk{1-\alpha}\rho\nk{\bar{Y}}\) for any \(\bar{X},\bar{Y}\in\productspace{}\) and any \(\alpha\in\zeroone{}\).
         \item\label{sysconvexrisk} Risk convexity: Suppose \(\rho\nk{\bar{Z}\nk{\omega}}=\alpha\rho\nk{\bar{X}\nk{\omega}}+\nk{1-\alpha}\rho\nk{\bar{Y}\nk{\omega}}\) for given \(\alpha\in\zeroone{}\) and for almost all \(\omega\in\Omega{}\). Then \(\rho\nk{\bar{Z}}\leq\alpha\rho\nk{\bar{X}}+\nk{1-\alpha}\rho\nk{\bar{Y}}\).
        \end{enumerate}
  \item\label{sysposhom} Positive homogeneity: \(\rho\nk{\alpha\bar{X}}=\alpha\rho\nk{\bar{X}}\) for any \(\bar{X}\in\productspace{}\) and any \(\alpha\in\bR_{+}\).
  \item\label{sysprefcons} Preference consistency: If \(\rho\nk{\bar{X}\nk{\omega}}\geq\rho\nk{\bar{Y}\nk{\omega}}\) for almost all \(\omega\in\Omega{}\), then \(\rho\nk{\bar{X}}\geq\rho\nk{\bar{Y}}\).
  \item\label{sysrsurjectivity} \(\cR{}\)-Surjectivity: \(\rho\nk{\bR^{n}}=\cR{}\) for \(\cR=\bR{}\) or \(\cR=\bR_{+}\).
  \item\label{sysnormal} Normalization: \(\rho\nk{1_n}=n\).
 \end{enumerate}

 If a mapping satisfies the properties~\ref{sysmonoton} and~\ref{sysconvexoutcome} it is called systemic risk measure. If a systemic risk measure additionally satisfies~\ref{sysposhom} it is called positive homogeneous systemic risk measure.
\end{defi}

The properties~\ref{sysmonoton} and~\ref{sysconvexoutcome} are the most natural. Their interpretation is completely analogously to the corresponding properties of single-firm risk measures. The properties~\ref{sysconvexrisk} and~\ref{sysprefcons} were first introduced by~\cite{chen2013axiomatic} and are essential (together with~\ref{sysrsurjectivity} and additional properties) for the decomposition result of systemic risk measures of type~\eqref{eq:syskromer}. However, it turns out that these properties are strongly connected to the dual representation of the systemic risk measure. In relevant situations systemic risk measures of type~\eqref{eq:sysbiaginiran} also satisfy these properties.


\subsection{First Aggregate}\label{subsec:firstaggregate}
If we consider a mapping \(\rho\colon\productspace\to\bR\cup\gk{\infty}\) which satisfies certain properties of~\ref{def:sysriskproperties} it appears to be a systemic risk measure of type~\eqref{eq:syskromer}, i.e.\@ it can be understood as the composition of a single-firm risk measure and a general AR\@. For this purpose let us recall the definition of a single-firm risk measure.

\begin{defi}\label{def:singleriskproperties}
 A mapping \(\rhosing\colon\cX\to\bR\cup\gk{\infty}\) is called convex single-firm risk measure if it satisfies the following properties:
 \begin{enumerate}[label = (R\arabic*)]
  \item\label{singlemonoton} Monotonicity: If \(X\succcurlyeq{Y}\), then \(\rhosing\nk{X}\geq\rhosing\nk{Y}\) for any \(X,Y\in\cX{}\).
  \item\label{singleconvex} Convexity: \(\rhosing\nk{\alpha{X}+\nk{1-\alpha}Y}\leq\alpha\rhosing\nk{X}+\nk{1-\alpha}\rhosing\nk{Y}\) for any \(X,Y\in\cX{}\) and any \(\alpha\in\zeroone{}\).
 \end{enumerate}
 A positively homogeneous single-firm risk measure is a convex single-firm risk measure that additionally satisfies the property
 \begin{enumerate}[label = (R\arabic*)]
  \setcounter{enumi}{2}
  \item\label{singleposhom}Positive homogeneity: \(\rhosing\nk{\alpha{X}}=\alpha\rhosing\nk{X}\) for any \(X\in\cX{}\) and any \(\alpha\in\bR_{+}\).
 \end{enumerate}
 A coherent single-firm risk measure is a positively homogeneous single-firm risk measure that additionally satisfies the property
 \begin{enumerate}[label= (R\arabic*)]
  \setcounter{enumi}{3}
  \item\label{singletranslation} Translation property: \(\rhosing\nk{X+m}=\rhosing\nk{X}+m\) for any \(X\in\cX{}\) and any \(m\in\bR{}\).
 \end{enumerate}
\end{defi}
It is well known (see for example~\cite{follmer2008stochastic} Exercise 4.1.3), that under the assumption of positive homogeneity~\ref{singleposhom}, convexity~\ref{singleconvex} is equivalent to
\emph{
 \begin{enumerate}[label= (R\arabic*)]
  \setcounter{enumi}{4}
  \item\label{singlesubadd} Subadditivity: \(\rhosing\nk{X+Y}\leq\rhosing\nk{X}+\rhosing\nk{Y}\) for any \(X,Y\in\cX{}\).
 \end{enumerate}
}
Another property, which was introduced by~\cite{frittelli2002putting}, is
\emph{
 \begin{enumerate}[label = (R\arabic*)]
  \setcounter{enumi}{5}
  \item\label{singleconstancy} Constancy on \(\cR\subset\bR{}\): \(\rhosing\nk{m}=m\) for any \(m\in\cR{}\).
 \end{enumerate}
}

The only non standard property is~\ref{singleconstancy}. As already mentioned, this property is of technical nature and plays an important role for the upcoming decomposition result. It was introduced by~\cite{kromer2016systemic}. A special case is constancy on \(\gk{1}\) which is often referred to as normalization. This property follows from the translation property~\ref{singletranslation} together with positive homogeneity~\ref{singleposhom} or together with the property \(\rhosing\nk{0}=0\). The (economic) interpretations of the above properties are well known and for a detailed study we refer the reader to~\cite{follmer2008stochastic} and~\cite{frittelli2002putting}.

\begin{theo}\label{theo:sysriskcomposition}
 Consider a mapping \(\rho\colon\productspace\to\bR\cup\gk{\infty}\) which satisfies \(\rho\vert_{\bR^{n}}\nk{\productspace}=\cX{}\) or \(\rho\vert_{\bR^{n}}\nk{\productspace}=\cX_{+}\) for \(\cR=\bR_{+}\) (resp. \(\rho\vert_{\bR^{n}}\nk{\productspace}\subseteq\cX{}\) or \(\rho\vert_{\bR^{n}}\nk{\productspace}\subseteq\cX_{+}\) for \(\cR=\bR_{+}\)).
 Then \(\rho{}\) is a systemic risk measure which additionally satisfies the properties~\ref{sysconvexrisk},~\ref{sysprefcons} and~\ref{sysrsurjectivity} (resp.~\ref{sysconvexrisk},~\ref{sysposhom},~\ref{sysprefcons} and~\ref{sysnormal}) if and only if there exists a convex (resp.\@ positively homogeneous) AR \(\Lambda:\bR^{n}\to\bR{}\) and a convex (resp.\@ positively homogeneous) single firm risk measure \(\rhosing\colon\cX\to\bR\cup\gk{\infty}\) that satisfies the constancy property~\ref{singleconstancy} on \(\cR=\Lambda\nk{\bR^{n}}\),
 such that \(\rho{}\) is the composition of \(\rhosing{}\) and \(\Lambda{}\) for all \(\bar{X}\in\productspace{}\), i.e.
 \begin{equation}\label{eq:sysriskcomposition}
  \rho\nk{\bar{X}}=\nk{\rhosing\circ\Lambda}\nk{\bar{X}}.
 \end{equation}
\end{theo}

\begin{proof}
 See~\cite{kromer2016systemic} Theorem 3.1 and Corollary 3.2.
\end{proof}

Theorem~\ref{theo:sysriskcomposition} can be seen as a blueprint for the construction of systemic risk measures of type~\eqref{eq:syskromer}. Additionally, if a given mapping satisfies the desired properties the proof yields an instruction to find \(\rhosing{}\) and \(\Lambda{}\). Let us continue by presenting some examples.

\begin{ex}\label{ex:sysentropic}
  Consider a setting in which the regulator uses the exponential loss function to express his preference. The suitable single firm risk measure  \(\rhosing{}\) in this setup is clearly the entropic risk measure given by
  \begin{equation}\label{eq:signleentropic}
   \rhosing^{entr}\nk{X}
   =\frac{1}{\theta{}}\ln\bE\ek{\exp\nk{\theta{}X}},\qquad\theta{}>0.
  \end{equation}
  The suitable domain for this risk measure is the Orlicz heart \(\Mexp{}\) with the Young function \(\phi^{\exp}\nk{x}:=\exp\nk{\abs{x}}-1\) (see~\cite{follmer2015axiomatic}). Note that in this example the existence of a base probability measure \(\bP{}\) is assumed. The parameter \(\theta{}\) reflects the risk aversion of the operator, smaller \(\theta{}\) means lower risk aversion. If we use the entropic risk measure as a building block for a systemic risk measure, \(\theta{}\) can be understood as a systemic risk aversion parameter for the whole system. It gives the regulator the opportunity to adjust the risk measure. A suitable way to derive such a systemic risk aversion parameter is to simply assign an individual risk aversion parameter \(\alpha_{i}\) to each participant of the system. Note that \(\alpha_{i}\) reflects the risk aversion from the point of view of the regulator, i.e.\ risk seeking participants will receive bigger \(\alpha{}\). Now the systemic risk aversion parameter has to act contrary to the risk aversion of the system and its participants. Risk averse participants of the system give the regulator the opportunity to be more flexible with the systemic risk measure, whereas risk friendly participants yield a conservative systemic risk measurement. Therefore, we set
  \[
    \theta=\frac{1}{\sum_{i=1}^{n}\frac{1}{\alpha_{i}}}.
  \]
  We are now able to construct different systemic risk measures by choosing different ARs\@. If we choose the ARs presented in Example~\ref{ex:arlossesonly}, we obtain for \(\Lambda^{\text{sum}}\)
  \begin{equation}\label{eq:sysentropicsum}
   \rho^{\text{ses}}\nk{\bar{X}}=\nk{\rhosing^{entr}\circ\Lambda^{\text{sum}}}\nk{\bar{X}}=\frac{1}{\theta}\ln\bE\ek{\exp\nk{\theta\sum_{i=1}^{n}X_{i}}},
  \end{equation}
  for \(\Lambda^{\text{sum,c}}\)
  \begin{equation}\label{eq:sysentropicsumwithshift}
   \begin{split}
     \rho^{\text{ses,c}}\nk{\bar{X}}
     &=\nk{\rhosing^{entr}\circ\Lambda^{\text{sum,c}}}\nk{\bar{X}}=\frac{1}{\theta}\ln\bE\ek{\exp\nk{\theta\nk{\sum_{i=1}^{n}X_{i}-c}}}\\
     &=\frac{1}{\theta}\ln\bE\ek{\exp\nk{\theta\sum_{i=1}^{n}X_{i}}}-c=\nk{\rhosing^{entr}\circ\Lambda^{\text{sum}}}\nk{\bar{X}}-c,
   \end{split}
  \end{equation}
  and for \(\Lambda^{\text{loss}}\)
  \begin{equation}\label{eq:sysentropicloss}
   \rho^{\text{sel}}\nk{\bar{X}}=\nk{\rhosing^{entr}\circ\Lambda^{\text{loss}}}\nk{\bar{X}}=\frac{1}{\theta}\ln\bE\ek{\exp\nk{\theta\sum_{i=1}^{n}X_{i}^{+}}}.
  \end{equation}
 \end{ex}

\begin{ex}\label{ex:biaginiinkromersetting}
 In this example we use the same AR and the single-firm risk measure corresponding to the acceptance set presented in~\cite{biagini2020fairness}, section 6. They are given as follows:
 \begin{align*}
  \Lambda^{\text{exut}}\nk{\bar{x}}=\sum_{i=1}^{n}\frac{1}{\alpha_{i}}\exp\nk{\alpha_{i}x_{i}} &  & \text{and} &  & \rhosing\nk{X}=\rho_{\cA}\nk{X}=\inf\gk{m\in\bR\mid{X}-m\in\cA},
 \end{align*}
 where the set \(\cA{}\) is given by
 \[
  \cA=\gk{Z\in\Lone\probspace{}\mid\bE\ek{Z}\leq{B}}
 \]
 for a constant \(B<\infty{}\) and \(\alpha_{i}>0\), \(i\in\gk{\oneton{}}\). Obviously, \(\rhosing{}\) is a convex monetary risk measure. Moreover, it is coherent if \(B=0\), and it has the representation
 \[
  \rhosing\nk{X}=\bE\ek{X}-B.
 \]
 The corresponding systemic risk measure is given by
 \begin{equation}\label{eq:biaginiinkromersetting}
  \rho\nk{\bar{X}}=\sum_{i=1}^{n}\frac{1}{\alpha_{i}}\bE\ek{\exp\nk{\alpha_{i}X_{i}}}-B
 \end{equation}
\end{ex}

As for single-firm risk measures, primal and dual representation results are of major interest. In this context the epigraphs of \(\rhosing{}\) and \(\Lambda{}\), given by
\begin{align}
 \cA_{\rhosing} 
 &=\gk{\nk{m,X}\in\bR\times\cX\mid{m}\geq\rhosing\nk{X}}\label{eq:sysacceptancerho},                \\
 \cA_{\Lambda} 
 &=\gk{\nk{Y,\bar{Z}}\in\cX\times\productspace\mid{Y}\succcurlyeq\Lambda(\bar{Z})},\label{eq:sysacceptancelambda}
\end{align}
are of major interest. For systemic risk measures \(\rho=\rhosing\circ\Lambda{}\) with a convex single firm risk measure \(\rhosing\colon\cX\to\bR\cup\gk{\infty}\) that satisfies the constancy property~\ref{singleconstancy} on \(\cR=\Lambda\nk{\bR^{n}}\) and a convex AR \(\Lambda\colon\bR^{n}\to\bR{}\), these sets appear to have the monotonicity and the epigraph property.

\begin{defi}\label{def:monotonicityandepigraphproperty}
 Let \(V\times{W}\subseteq\cX\times\productspace{}\) be two linear spaces. A set \(\cS\subset{V}\times{W}\) satisfies the monotonicity property if \(\nk{v,w_1}\in\cS,\ w_2\in{W}\) and \(w_1\succcurlyeq{w}_2\) imply \(\nk{v,w_2}\in\cS{}\).
 A set \(\cS\subset{V}\times{W}\) satisfies the epigraph property \(\nk{v_1,w}\in\cS,\ v_2\in{V}\) and \(v_2\succcurlyeq{v}_1\) imply \(\nk{v_2,w}\in\cS{}\).
\end{defi}

\begin{prop}\label{prop:sysprimal}
 Suppose that \(\rho=\rhosing\circ\Lambda{}\) is a systemic risk measure with convex single firm risk measure \(\rhosing:\cX\to\bR\cup\gk{\infty}\) that satisfies the constancy property~\ref{singleconstancy} on \(\cR=\Lambda(\bR^n)\) and convex AR \(\Lambda:\bR^n\to\bR{}\).
 Denote by \(\cA_{\rhosing}\) (resp. \(\cA_{\Lambda}\)) the corresponding epigraphs.
 \begin{compactenum}[(i)]
  \item The sets \(\cA_{\rhosing}\) and \(\cA_{\Lambda}\) satisfy the following properties:
  \begin{compactenum}[(a)]
   \item \(\cA_{\rhosing}\) and \(\cA_{\Lambda}\) satisfy the monotonicity property.
   \item \(\cA_{\rhosing}\) and \(\cA_{\Lambda}\) satisfy the epigraph property.
   \item \(\cA_{\rhosing}\) and \(\cA_{\Lambda}\) are convex sets.
   \item \(\nk{a,a}\in\cA_{\rhosing}\) with \(\inf\gk{r\in\bR:\nk{r,a}\in\cA_{\rhosing}}=a\) for all \(a\in\text{Im}\Lambda{}\)
  \end{compactenum}
  If \(\rho=\rhosing\circ\Lambda{}\) is a positively homogeneous systemic risk measure, then the following properties are additionally satisfied:
  \begin{compactenum}[(a)]
   \setcounter{enumii}{4}
   \item \(\cA_{\rhosing}\) and \(\cA_{\Lambda}\) are convex cones.
  \end{compactenum}
  \item For any \(\bar{X}\in\productspace{}\), \(\rho{}\) admits the primal representation
  \begin{equation}\label{eq:sysprimal}
   \rho\nk{\bar{X}}=\inf\gk{m\in\bR\mid\nk{m,Y}\in\cA_{\rhosing},\nk{Y,\bar{X}}\in\cA_{\Lambda}}
  \end{equation}
  where we set \(\inf\emptyset:=\infty{}\).
 \end{compactenum}
\end{prop}

\begin{proof}
 See~\cite{kromer2016systemic} Proposition 4.2.
\end{proof}

The next step is to use the primal representation and find its corresponding dual problem. This technique yields the so-called dual representation.

\begin{theo}\label{theo:sysdualconvex}
  Suppose that \(\rho=\rhosing\circ\Lambda{}\) is a systemic risk measure with a lsc convex single firm risk measure \(\rhosing:\cX\to\bR\cup\gk{\infty}\) that satisfies the constancy property~\ref{singleconstancy} on \(\cR=\Lambda(\bR^n)\) and a convex AR \(\Lambda:\bR^n\to\bR{}\) that is continuous on \(\productspace{}\). Then for any \(\bar{X}\in\productspace{}\)
 \begin{equation}\label{eq:sysdualconvex}
  \rho\nk{\bar{X}}=\sup_{\xixi\in\cD}\gk{\npairing{\bar{X}}{\Xi}-\alpha\xixi},
 \end{equation}
 where \(\alpha\colon\cX^{\prime}\times\dualproductspace\to\bR\cup\gk{\infty}\) is given by
 \begin{equation}
  \begin{split}
   \alpha\xixi
   &=\alpha^{\rhosing}\nk{\xi}+\alpha^{\Lambda}\xixi{}\\
   &=\sup_{\nk{m,Y}\in\cA_{\rhosing}}\gk{-m+\pairing{Y}{\xi}}+\sup_{\nk{V,\bar{Z}}\in\cA_{\Lambda}}\gk{-\pairing{V}{\xi}+\npairing{\bar{Z}}{\Xi}}\\
   &=\sup_{\nk{m,Y}\in\cA_{\rhosing},\nk{V,\bar{Z}}\in\cA_{\Lambda}}\gk{-m+\pairing{Y-V}{\xi}+\npairing{\bar{Z}}{\Xi}}
  \end{split}
 \end{equation}
 and \(\cD:=\gk{\xixi\in\cX^{\prime}\times\dualproductspace\mid\alpha\xixi<\infty}\). In addition, feasible solutions \(\xixi{}\) of the optimization problem~\eqref{eq:sysdualconvex} are non-negative and the \(\xi{}\)-component of a feasible solution satisfies \(\pairing{1}{\xi}=1\) in case of \(\rho\nk{\bR^n}=\bR{}\)
 and \(\pairing{1}{\xi}\leq1\) in case of \(\rho\nk{\bR^n}=\bR_+\).
\end{theo}

\begin{proof}
 See~\cite{kromer2016systemic} Theorem 4.3.
\end{proof}

\begin{rem}\label{rem:specialcasesforthedualrepresentation}
  In the situation where at least one of the components, \(\rhosing{}\) or \(\Lambda{}\), is positively homogeneous we are able to concretize the situations where we have \(\alpha\xixi<\infty{}\). For this purpose let 

\begin{align*}
  \cA^{\prime}_{\rhosing}&:=\gk{\nk{x,\psi}\in\bR\times\cX^{\prime}\mid{mx}-\pairing{Y}{\psi}\geq0\forall\nk{m,Y}\in\cA_{\rhosing}},\\
  \cA^{\prime}_{\Lambda}&:=\gk{\xixi\in\cX^{\prime}\times\dualproductspace\mid\pairing{Y}{\xi}-\npairing{\bar{Z}}{\Xi}\geq0\forall\nk{Y,\bar{Z}}\in\cA_{\Lambda}}.
\end{align*}

Up to a sign change, these sets are the dual cones to \(\cA_{\rhosing}\) and \(\cA_{\Lambda}\). Now, if \(\rhosing{}\) is positively homogeneous we have 
\[
  \alpha^{\rhosing}\nk{\xi}=
  \begin{cases}
    0,&\nk{1,\xi}\in\cA^{\prime}_{\rhosing},\\
    \infty,&\nk{1,\xi}\not\in\cA^{\prime}_{\rhosing},
  \end{cases}
\]
and if \(\Lambda{}\) is positively homogeneous we have 
\[
  \alpha^{\Lambda}\xixi=
  \begin{cases}
    0,&\xixi\in\cA^{\prime}_{\Lambda},\\
    \infty,&\xixi\not\in\cA^{\prime}_{\Lambda}.
  \end{cases}
\]
\end{rem}

Consequently, the observations in Remark~\ref{rem:specialcasesforthedualrepresentation} yield to a simplified dual representation for positively homogeneous systemic risk measures.

\begin{theo}\label{theo:sysdualposhom}
 Suppose that \(\rho=\rhosing\circ\Lambda{}\) is a positively homogeneous systemic risk measure with a lsc positively homogeneous single-firm risk measure \(\rhosing{}\) that satisfies constancy on \(\Lambda\nk{\bR^{n}}\) and a positively homogeneous AR \(\Lambda{}\) that is continuous on \(\productspace{}\). Then for any \(\bar{X}\in\productspace{}\)
 \begin{equation}\label{eq:sysdualposhom}
  \rho\nk{\bar{X}}=\sup_{\cV^{\#}}\npairing{\bar{X}}{\Xi}
 \end{equation}
 where the set \(\cV^{\#}\) is defined by
 \begin{equation}\label{eq:sysdualposhomvhash}
  \cV^{\#}:=\gk{\xixi\in\cX^{\prime}\times\dualproductspace\mid\nk{1,\xi}\in\cA^{`}_{\rhosing},\xixi\in\cA^{`}_{\Lambda}}.
 \end{equation}
 In addition, feasible solutions \(\xixi{}\) of the optimization problem~\eqref{eq:sysdualposhom} are non-negative and the \(\xi{}\)-component of a feasible solution satisfies \(\pairing{1}{\xi}=1\) in case of \(\rho\nk{\bR^n}=\bR{}\)
 and \(\pairing{1}{\xi}\leq1\) in case of \(\rho\nk{\bR^n}=\bR_+\) and the \(\Xi{}\)-component satisfies \(\npairing{1_n}{\Xi}\leq{n}\).
\end{theo}

\begin{proof}
 See~\cite{kromer2016systemic} Theorem 4.7.
\end{proof}

The following theorem and its corollary describe the relation between the dual representation of a positively homogeneous systemic risk measure and its directional derivative.

\begin{theo}\label{theo:sysdirectionalforposhom}
  Suppose that \(\rho=\rhosing\circ\Lambda{}\) is a positively homogeneous systemic risk measure with a lsc positively homogeneous single-firm risk measure \(\rhosing{}\) that satisfies constancy on \(\Lambda\nk{\bR^{n}}\) and a positively homogeneous AR \(\Lambda{}\) that is continuous on \(\productspace{}\). Then the directional derivative of \(\rho{}\) at \(\bar{X}\) in the direction of \(\bar{Y}\) exists and is given by
 \begin{equation}
  \delta_{+}\rho\nk{\bar{X},\bar{Y}}=\max_{\xixi\in\cV^{\#}(\bar{X})}\npairing{\bar{Y}}{\Xi},
 \end{equation}
 where the set \(\cV^{\#}(\bar{X})\) is defined by
 \[
  \cV^{\#}(\bar{X}):=\gk{\xixi\in\cV^{\#}\mid\rho\nk{\bar{X}}=\npairing{\bar{X}}{\Xi}}.
 \]
\end{theo}

\begin{proof}
 See~\cite{kromer2016systemic} Theorem 5.2.
\end{proof}

\begin{cor}\label{cor:sysdirectionalforposhomunique}
  Suppose that \(\rho=\rhosing\circ\Lambda{}\) is a positively homogeneous systemic risk measure with a lsc positively homogeneous single-firm risk measure \(\rhosing{}\) that satisfies constancy on \(\Lambda\nk{\bR^{n}}\) and a positively homogeneous AR \(\Lambda{}\) that is continuous on \(\productspace{}\). If the optimal solution \(\xixiopt{}\) for the dual problem~\eqref{eq:sysdualposhom} is unique, then \(\rho{}\) is Gâteaux differentiable with derivative \(\Xiopt{}\) at \(\bar{X}\), i.e.,
 \begin{equation}
  \delta\rho\nk{\bar{X},\bar{Y}}=\npairing{\bar{Y}}{\Xiopt{}}.
 \end{equation}
\end{cor}

\begin{proof}
 This is a direct consequence of Theorem~\ref{theo:sysdirectionalforposhom} and the fact that
 \[
  \delta_{-}\rho\nk{\bar{X},\bar{Y}}=-\delta_{+}\rho\nk{\bar{X},-\bar{Y}}.
 \]
\end{proof}

More general, for systemic risk measures \(\rho=\rhosing\circ\Lambda{}\) which are not necessarily Gâteaux differentiable, there is a connection between optimal solutions \(\xixiopt{}\) to the dual problem~\eqref{eq:sysdualposhom} (resp.~\eqref{eq:sysdualconvex}) and the subgradient.

\begin{cor}\label{cor:syssubgradientforconvex}
  Suppose that \(\rho=\rhosing\circ\Lambda{}\) is a systemic risk measure with a lsc convex single firm risk measure \(\rhosing:\cX\to\bR\cup\gk{\infty}\) that satisfies the constancy property~\ref{singleconstancy} on \(\cR=\Lambda(\bR^n)\) and a convex AR \(\Lambda:\bR^n\to\bR{}\) that is continuous on \(\productspace{}\) and fix \(\bar{X}\in\productspace{}\).
 Then for any optimal solution \(\xixiopt{}\) to the dual problem~\eqref{eq:sysdualconvex}, \(\Xiopt{}\) is a subgradient of \(\rho{}\) at \(\bar{X}\), i.e., \(\Xiopt\in\partial\rho\nk{\bar{X}}\).
\end{cor}

\begin{proof}
 See~\cite{kromer2016systemic} Corollary 4.9.
\end{proof}

\begin{ex}\label{ex:dualrepsysentropic}
  Consider the systemic risk measure \(\rho^{ses}\) presented in Example~\ref{ex:sysentropic}. Obviously \(\Lambda^{sum}\) is Gâteaux differentiable. Moreover, we have that for all \(X,V\in\Mexp{}\)
  \begin{align*}
    \delta\rhosing^{entr}\nk{X,V}
    &=\frac{\bE\ek{V\exp\nk{\theta{X}}}}{\bE\ek{\exp\nk{\theta{X}}}}\\
    &=\pairing{V}{\frac{\exp\nk{\theta{X}}}{\bE\ek{\exp\nk{\theta{X}}}}}.
  \end{align*}
  But this means \(\rhosing^{entr}\) is Gâteaux differentiable with derivative 
  \[
    \nabla\rhosing^{entr}\nk{X}=\frac{\exp\nk{\theta{X}}}{\bE\ek{\exp\nk{\theta{X}}}}
  \]
  at \(X\in\Mexp{}\). Now Proposition~\ref{prop:rulesforgateaux} yields, that \(\rho^{ses}\) is also Gâteaux differentiable with derivative
  \[
    \nabla\rho^{ses}\nk{\bar{X}}=1_{n}\frac{\exp\nk{\theta\sum_{i=1}^{n}X_{i}}}{\bE\ek{\exp\nk{\theta\sum_{i=1}^{n}X_{i}}}}=1_{n}\nabla\rhosing^{entr}\nk{\sum_{i=1}^{n}X_{i}}.
  \] 
  From the previous corollary we know that for optimal solutions \(\xixiopt{}\) to the  dual problem the component \(\Xiopt{}\) is a subgradient. For \(\rho^{ses}\) the subgradient is a singleton consisting of \(\nabla\rho^{ses}\nk{\bar{X}}\) since \(\rho^{ses}\) is continuous on \(\nk{\Mexp}^{n}\). Since \(\Lambda^{sum}\) is a positively homogeneous AR only \(\xixi\in\cA^{\prime}_{\Lambda^{sum}}\) are of interest. Additionally, we have \(\Xi=1_{n}\xi{}\) and therefore an optimal solution to the dual problem is given by  
  \[
    \xixiopt{}=\nk{\nabla\rhosing^{entr}\nk{\sum_{i=1}^{n}X_{i}},1_{n}\nabla\rhosing^{entr}\nk{\sum_{i=1}^{n}X_{i}}}.
  \]
  Note, that this is also an optimal solution to the dual problem of \(\rho^{ses,c}\) presented in Example~\ref{ex:sysentropic}.
\end{ex}


\subsection{First Inject Capital}\label{subsecfirstinjectcapital}

If we want to consider systemic risk measures of type~\eqref{eq:sysbiaginidet} and~\eqref{eq:sysbiaginiran} the procedure is a little different. Systemic risk measures of type~\eqref{eq:syskromer} arise in an axiomatic approach. The axioms are motivated by economical conditions. If one is interested in mappings which fulfill a certain combination of these axioms they are the proper choice. In contrast, systemic risk measures of type~\eqref{eq:sysbiaginiran} appear as given mappings, and we have to show that they satisfy desirable properties. The suitable setting for these risk measures is an Orlicz space setting which is described in~\ref{appendix:orliczspaces}. Since these spaces are Banach lattices (equipped with the Luxemburg norm) Definition~\ref{def:sysriskproperties} already fits to define the properties we are interested in. Let us start by formulating the standing assumptions for our analysis, which were introduced by~\cite{biagini2020fairness} to derive dual representation results.

\begin{ass*}
 \begin{compactenum}[(i)]
  \item \(\cC_{0}\subseteq\cC\nk{\bR}\) and \(\cC=\cC_{0}\cap\MPhi{}\) is a convex cone which satisfies \(\bR^{n}\subseteq\cC{\subseteq\cC\nk{\bR}}\).
  \item \(\Lambda\colon\bR^{n}\to\bR{}\) given by \(\Lambda\nk{\bar{x}}=\sum_{i=1}^{n}l_{i}\nk{x_{i}}\), where \(l_{i}:\bR\to\bR{}\) is increasing, strictly convex, differentiable and satisfies the Inada conditions
  \[
   l_{i}^{\prime}\nk{\infty}:=\lim_{x\to\infty}l_{i}^{\prime}\nk{x}=\infty{},\quad{l}_{i}^{\prime}\nk{-\infty}:=\lim_{x\to-\infty}l_{i}^{\prime}\nk{x}=0.
  \]
  \item \(B>\Lambda\nk{-\infty}\), i.e.\@ there exists \(\bar{m}\in\bR^{n}\) such that \(\Lambda\nk{\bar{m}}=\sum_{i=1}^{n}l_{i}\nk{m_i}\leq{B}\).
  \item \(\cA_{0}:=\gk{{Z}\in\Lone\probspace{}\mid\bE\ek{{Z}}\leq{B}}\).
  \item For all \(i\in\gk{\oneton{}}\), it holds that for any probability measure \(\bQ\ll\bP{}\) with density \(\frac{d\bQ}{d\bP}:=\xi{}\) that
  \[
   \bE\ek{l_{i}^{*}\nk{\xi}}<\infty \quad \text{if and only if} \quad \bE\ek{l_{i}^{*}\nk{\lambda\xi}}<\infty,\text{ for all }\lambda>0,
  \]
  where \(l_{i}^{*}\nk{y}:=\sup_{x\in\bR}\gk{xy-l_{i}\nk{x}}\).
 \end{compactenum}
\end{ass*}

Under these assumptions we analyze mappings \(\brisk\colon\MPhi\to\bR\cup\gk{-\infty}\cup\gk{\infty}\) given by

\begin{align}\label{eq:biaginisysunderassumption}
 \brisk\nk{\bar{X}}
 &=\inf\gk{\sum_{i=1}^{n}Y_{i}\mid\bar{Y}\in\cC,\Lambda\nk{\bar{X}-\bar{Y}}\in\cA_{0}} \nonumber{}\\
 &=\inf\gk{\sum_{i=1}^{n}Y_{i}\mid\bar{Y}\in\cC,\bE\ek{\sum_{i=1}^{n}l_{i}\nk{X_i-Y_i}}\leq{B}}.
\end{align}

\begin{rem}\label{rem:choiceofb}
 A reasonable choice for \(B\) is given by \(\sum_{i=1}^{n}l_{i}\nk{0}\). It ensures that \(\brisk\nk{0}=0\) which is a desirable property in some situations as we will see later.
\end{rem}

\begin{prop}\label{prop:biaginisysproperties}
 Consider the mapping in~\eqref{eq:biaginisysunderassumption}.
 \begin{compactenum}[(i)]
  \item \(\brisk\nk{\bar{X}}>-\infty{}\), for all \(\bar{X}\in\MPhi{}\).
  \item \(\brisk{}\) is a systemic risk measure on \(\MPhi{}\), i.e.\@ it satisfies~\ref{sysmonoton} and~\ref{sysconvexoutcome}.
  \item \(dom\nk{\brisk}=\MPhi{}\), i.e., \(\brisk\colon\MPhi\to\bR{}\).
  \item \(\brisk{}\) is continuous and sub-differentiable on \(dom\nk{\brisk}=\Mphi{}\).
 \end{compactenum}
\end{prop}

\begin{proof}
 See~\cite{biagini2020fairness} Proposition 2.4.
\end{proof}

We continue by presenting one relevant example for the set \(\cC{}\).
\begin{ex}\label{ex:cgroups}
 Suppose that the participants of the financial system are divided into \(h\in\gk{1,\ldots,n}\) groups. This means, if we set \(\bar{n}=\nk{{n}_{1},\ldots,{n}_{h}}\in\bN^{h}\) with \({n}_{j-1}<{n}_{j}\), \(j=1\ldots,h\), \({n}_{0}:=0\) and \({n}_{h}:=n\), group \(j\) consists of the firms \(I_{j}:=\gk{{n}_{j-1}+1,\ldots,{n}_{j}}\) for \(j=1\ldots,h\).
 Consider the set \(\cC^{\nk{\bar{n}}}=\cC^{\nk{\bar{n}}}_{0}\cap\MPhi{}\), where
 \begin{equation}\label{eq:cgroups}
  \cC^{\nk{\bar{n}}}_{0}=\gk{\bar{Y}\in\Lzeron\mid\exists\bar{d}\in\bR^{h}:\sum_{i\in{I}_{j}}{Y}_i={d}_{j}\text{ for }j=1,\ldots,h}\subseteq\cC\nk{\bR}.
 \end{equation}
 The random vectors in \(\cC^{\nk{\bar{n}}}\) set a deterministic value for the allocation to each of the \(h\) groups. Inside each group the allocation of this fixed value is dependent on the occurrent scenario. Obviously, there are two extreme cases, i.e. \(h=1\) and \(h=n\). The case \(h=1\) leads to arbitrary random allocations with the only constraint \(\bar{Y}\in\cC\nk{\bR}\). The case \(h=n\) leads to fully deterministic allocations.
\end{ex}

The next theorem presents the dual representation for systemic risk measures given by~\eqref{eq:biaginisysunderassumption}.

\begin{theo}[\cite{biagini2020fairness}]\label{theo:biaginisysdual}
 For any \(\bar{X}\in\MPhi{}\),
 \begin{equation}\label{eq:biaginisysdual}
  \brisk\nk{\bar{X}}=\max_{\Xi\in\cD}\gk{\npairing{\bar{X}}{\Xi}-\alpha_{\Lambda,B}\nk{\Xi}},
 \end{equation}
 where the penalty function is given by
 \begin{equation}\label{eq:penaltybiagini}
  \alpha_{\Lambda,B}\nk{\Xi}:=\sup_{\bar{Z}\in\cA}\gk{\npairing{\bar{Z}}{\Xi}}, \quad\Xi\in\cD,
 \end{equation}
 with \(\cA\colon=\gk{\bar{Z}\in\MPhi\mid\sum_{i=1}^{n}\bE\ek{l_{i}\nk{{Z}_{i}}}\leq{B}}\) and
 \begin{equation*}
  \cD:=dom\nk{\alpha_{\Lambda,B}}\cap\gk{\Xi\in{L}^{\Phi^{*}}_{+}\mid\pairing{1}{\Xi_{i}}=1\ \forall{i}\text{ and }\npairing{\bar{Y}}{\Xi}-\sum_{i=1}^{n}{Y}_{i}\leq0\ \forall\ \bar{Y}\in\cC}.
 \end{equation*}
 \begin{compactenum}[(i)]
  \item Suppose that for some \(i,j\in\gk{1,\ldots,n},\ i\not=j\) we have \(\pm\nk{e_{i}\ind{A}-e_j\ind{A}}\in\cC{}\) for all \(A\in\cF{}\). Then
  \begin{equation*}
   \cD:=dom\nk{\alpha_{\Lambda,B}}\cap\gk{\Xi\in{L}^{\Phi^{*}}_{+}\mid\pairing{1}{\Xi_{i}}=1\ \forall{i},\ \Xi_{i}=\Xi_{j}\text{ and }\npairing{\bar{Y}}{\Xi}-\sum_{i=1}^{n}{Y}_{i}\leq0\ \forall\ \bar{Y}\in\cC}.
  \end{equation*}
  \item Suppose that \(\pm\nk{e_{i}\ind{A}-e_j\ind{A}}\in\cC{}\) for all \(i,j\in\gk{1,\ldots,n}\) and all \(A\in\cF{}\). Then
  \begin{equation*}
   \cD:=dom\nk{\alpha_{\Lambda,B}}\cap\gk{\Xi\in{L}^{\Phi^{*}}_{+}\mid\pairing{1}{\Xi_{i}}=1, \Xi_{i}=\Xi_{j}\ \forall{i,j\in\gk{1,\ldots,n}}}.
  \end{equation*}
 \end{compactenum}
\end{theo}

\begin{proof}
 See~\cite{biagini2020fairness} Proposition 3.1.
\end{proof}

\begin{ex}\label{ex:exponentialcase}
 We set \(\cC=\cC^{\bar{n}}\) and \(l_{i}\nk{x}=\frac{1}{\alpha_{i}}\exp\nk{\alpha_{i}x}\), \(\alpha_{i}>0\), \(i\in\gk{1,\ldots,n}\). This means \(\phi_{i}\nk{x}=\frac{1}{\alpha_{i}}\nk{\exp\nk{\alpha_{i}\abs{x}}-1}\). Since for \(\phi^{\exp}\nk{x}=\exp\nk{\abs{x}}-1\) we have \(\phi_{i}\sim\phi^{\exp}\) for all \(i\in\gk{\oneton{}}\) and therefore
 \(\MPhi{}=\nk{\Mexp}^{n}\). Additionally, set \(B>\sum_{i=1}^{n}l_{i}\nk{-\infty}=0\). The systemic risk measure \(\brisk^{ex}:\nk{\Mexp{}}^{n}\to\bR{}\) becomes
 \begin{equation}\label{eq:exponentialcase}
  \brisk^{ex}\nk{\bar{X}}=\inf\gk{\sum_{i=1}^{n}{Y}_{i}\mid\bar{Y}\in\cC^{\bar{n}},\bE\ek{\sum_{i=1}^{n}\frac{1}{\alpha_{i}}\exp\nk{\alpha_{i}\nk{{X}_i-{Y}_i}}}={B}}.
 \end{equation}
 The optimal solution of the dual problem~\eqref{eq:biaginisysdual} is given by the vector of probability densities \(\Xiopt{}\), where for \(j\in\gk{1,\ldots,h}\) the components \(l\in{I}_{j}\) are given by
 \begin{equation}
  \Xiopt_{l}=\Xioptj:=\frac{\exp\nk{\theta_{j}\sum_{i\in{I}_{j}}{X_i}}}{\bE\ek{\exp\nk{\theta_{j}\sum_{i\in{I}_{j}}{X_i}}}},
 \end{equation}
 where \(\theta_{j}=\frac{1}{\sum_{i\in{I}_{j}}\frac{1}{\alpha_{i}}}\).
 Moreover, the systemic risk measure \(\brisk^{ex}\) is Gâteaux differentiable with derivative \(\Xiopt{}\), and we obtain
 \begin{align*}
  \delta\brisk^{ex}\nk{\bar{X},\bar{V}}
   & =\npairing{\bar{V}}{\Xiopt{}}                                        \\
   & =\sum_{j=1}^{h}\sum_{i\in{I}_{j}}\pairing{V_{i}}{\Xioptj}.
 \end{align*}
\end{ex}

The following theorem shows that the properties~\ref{sysconvexrisk} and~\ref{sysprefcons} are direct consequences of the dual representation of systemic risk measures of type~\eqref{eq:sysbiaginiran}.

\begin{theo}\label{theo:biaginisysnewproperties}
 Suppose that \(\pm\nk{e_{i}\ind{A}-e_j\ind{A}}\in\cC{}\) for all \(i,j\in\gk{1,\ldots,n}\) and all \(A\in\cF{}\). Then the systemic risk measure \(\brisk{}\) satisfies the properties~\ref{sysconvexrisk},~\ref{sysprefcons} and~\ref{sysrsurjectivity}.
\end{theo}
\begin{proof}
 First note that, since we assume that \(\pm\nk{e_{i}\ind{A}-e_j\ind{A}}\in\cC{}\) for all \(i,j\in\gk{1,\ldots,n}\) and all \(A\in\cF{}\), the set \(\cD{}\) only consists of vectors of probability densities \(\Xi{}\) with \(\Xi_{i}=\xi{}\) for all \(i\in\gk{1,\ldots,n}\). We denote the corresponding probability measure by \(\bQ{}\). 
 Suppose that
 \begin{equation*}
  \brisk\nk{\bar{Z}\nk{\omega}}=\alpha\brisk\nk{\bar{X}\nk{\omega}}+\nk{1-\alpha}\brisk\nk{\bar{Y}\nk{\omega}},\qquad\alpha\in\interval{0}{1}
 \end{equation*}
 for almost all \(\omega\in\Omega{}\). Since \(\brisk{}\) admits a dual representation given by~\eqref{eq:biaginisysdual} we obtain
 \begin{align*}
  \brisk\nk{\bar{Z}\nk{\omega}}
   & =\max_{\Xi\in\cD}\gk{\sum_{i=1}^{n}{Z}_{i}\nk{\omega}-\alpha_{\Lambda,B}\nk{\Xi}}                                                                                                                 \\
   & =\sum_{i=1}^{n}{Z}_{i}\nk{\omega}+\max_{\Xi\in\cD}\gk{-\alpha_{\Lambda,B}\nk{\Xi}}                                                                        \end{align*}
   and
   \begin{align*}                                      
   \alpha\brisk\nk{\bar{X}\nk{\omega}}+\nk{1-\alpha}\brisk\nk{\bar{Y}\nk{\omega}}                                                                                                                             
   & =\alpha\max_{\Xi\in\cD}\gk{\sum_{i=1}^{n}{X}_{i}\nk{\omega}-\alpha_{\Lambda,B}\nk{\Xi}}\\
   & \ \ \ 
   +\nk{1-\alpha}\max_{\Xi\in\cD}\gk{\sum_{i=1}^{n}{Y}_{i}\nk{\omega}-\alpha_{\Lambda,B}\nk{\Xi}} \\
   & =\alpha\sum_{i=1}^{n}{X}_{i}\nk{\omega}+\nk{1-\alpha}\sum_{i=1}^{n}{Y}_{i}\nk{\omega}+\max_{\Xi\in\cD}\gk{-\alpha_{\Lambda,B}\nk{\Xi}}.
 \end{align*}
 But this means
 \begin{equation*}
  \sum_{i=1}^{n}{Z}_{i}=\alpha\sum_{i=1}^{n}{X}_{i}+\nk{1-\alpha}\sum_{i=1}^{n}{Y}_{i}
 \end{equation*}
 \(\bP{}-\)almost surely and since \(\bQ\ll\bP{}\) also \(\bQ{}-\)almost surely. Therefore,
 \begin{align*}
  \brisk\nk{\bar{Z}}
   & =\max_{\Xi\in\cD}\gk{\pairing{\sum_{i=1}^{n}{Z}_i}{\xi}-\alpha_{\Lambda,B}\nk{\Xi}}                                                \\
   & =\max_{\Xi\in\cD}\gk{\pairing{\alpha\sum_{i=1}^{n}{X}_{i}+\nk{1-\alpha}\sum_{i=1}^{n}{Y}_{i}}{\xi}-\alpha_{\Lambda,B}\nk{\Xi}} \\
   & =\brisk\nk{\alpha\bar{X}+\nk{1-\alpha}\bar{Y}}                                                                                                \\
   & \leq\alpha\brisk\nk{\bar{X}}+\nk{1-\alpha}\brisk\nk{\bar{Y}},
 \end{align*}
 where the last inequality follows from~\ref{sysconvexoutcome}. But this means \(\brisk{}\) satisfies~\ref{sysconvexrisk}. To show preference constancy~\ref{sysprefcons} suppose that
 \begin{equation*}
  \brisk\nk{\bar{X}\nk{\omega}}\geq\brisk\nk{\bar{Y}\nk{\omega}}
 \end{equation*}
 for almost all \(\omega\in\Omega{}\). With the same arguments as stated above we obtain
 \begin{equation*}
  \sum_{i=1}^{n}{X}_{i}\geq\sum_{i=1}^{n}{Y}_{i}
 \end{equation*}
 \(\bP{}-\)almost surely and since \(\bQ\ll\bP{}\) also \(\bQ{}-\)almost surely. But this means
 \begin{align*}
  \brisk\nk{\bar{X}}
   & =\max_{\Xi\in\cD}\gk{\pairing{\sum_{i=1}^{n}{X}_i}{\xi}-\alpha_{\Lambda,B}\nk{\Xi}}    \\
   & \geq\max_{\Xi\in\cD}\gk{\pairing{\sum_{i=1}^{n}{Y}_i}{\xi}-\alpha_{\Lambda,B}\nk{\Xi}} \\
   & =\brisk\nk{\bar{Y}}.
 \end{align*}
 The \(\cR{}\)-surjectivity property~\ref{sysrsurjectivity} follows with the same arguments.
\end{proof}

\begin{rem}\label{rem:conditionsfordecompositionofbiagini}
  \begin{enumerate}
    \item If we want to apply Theorem~\ref{theo:sysriskcomposition}, we additionally need that \(\brisk\mid_{\bR^{n}}\nk{\MPhi{}}=\cX{}\) for some locally convex Hausdorff space \(\cX{}\). To guarantee this property we need additional information about the loss functions \(l_{i}\) and their corresponding Young functions \(\phi_{i}\). For example if \(\phi_{i}\succ\phi_{j}\) for all \(i\neq{j}\), then obviously \(\brisk\mid_{\bR^{n}}\nk{\MPhi{}}=M^{\phi_{j}}\).
    \item The crucial assumption in Theorem~\ref{theo:biaginisysnewproperties} is \(\pm\nk{e_{i}\ind{A}-e_j\ind{A}}\in\cC{}\) for all \(i,j\in\gk{1,\ldots,n}\) and all \(A\in\cF{}\). The economic interpretation of this property is that capital transfers between the firms is an accepted instrument for the regulator to ensure the stability of the system. Obviously, the task of capital allocation exactly works that way. A situation where the regulator does not have this instrument would imply that the subgroups are regulated in an isolated manner, e.g.\@ insurances are regulated ignoring what is happening in the banking sector. 
    \item Obviously, the assumption is always fulfilled if we set \(\cC=\cC\nk{\bR}\). According to the definition of the systemic risk measure \(R\), in the least conservative or least restrictive situation we are always able to apply Theorem~\ref{theo:biaginisysnewproperties}.
  \end{enumerate}
\end{rem}

\begin{prop}[~\cite{biagini2020fairness}] If \(\alpha_{\Lambda,B}\nk{\Xi}<\infty{}\), the penalty function in~\eqref{eq:penaltybiagini} can be written as
  \begin{equation}\label{eq:alternativerepresentationpenaltybiagini}
    \alpha_{\Lambda,B}\nk{\Xi}:=\sup_{\bar{Z}\in\cA}\gk{\npairing{\bar{Z}}{\Xi}}=\inf_{\lambda>0}\nk{\frac{1}{\lambda}B+\frac{1}{\lambda}\sum_{i=1}^{n}\bE\ek{l_{i}^{*}\nk{\lambda\Xi_{i}}}}
  \end{equation}
and \(\bE\ek{l_{i}^{*}\nk{\lambda\Xi_{i}}}<\infty{}\) for all \(i\) and all \(\lambda>0\). Additionally, the infimum is attained in~\eqref{eq:alternativerepresentationpenaltybiagini}, i.e.,
\begin{equation}
  \begin{split}
  \alpha_{\Lambda,B}\nk{\Xi}
  &=\sum_{i=1}^{n}\pairing{\nk{l_{i}^{*}}^{\prime}\nk{\lambda^{*}\Xi_{i}}}{\Xi_{i}}\\
  &=\sum_{i=1}^{n}\bE\ek{\Xi_{i}\nk{l_{i}^{*}}^{\prime}\nk{\lambda^{*}\Xi_{i}}},
\end{split}
\end{equation}
where \(\lambda^{*}>0\) is the unique solution of the equation
\begin{equation}
  B+\sum_{i=1}^{n}\bE\ek{l_{i}^{*}\nk{\lambda\Xi_{i}}}-\lambda\sum_{i=1}^{n}\bE\ek{\Xi_{i}\nk{l_{i}^{*}}^{\prime}\nk{\lambda\Xi_{i}}}=0.
\end{equation}
\end{prop}
\begin{proof}
  See~\cite{biagini2020fairness} Proposition 3.4.
\end{proof}
\begin{ex}\label{ex:lambdastarandpenaltyexponentialcase}
  Let \(l_{i}\nk{x}=\frac{1}{\alpha_{i}}\exp\nk{\alpha_{i}x},\ \alpha_{i}>0,\ i\in\oneton{}\), as in Example~\ref{ex:exponentialcase}.
  We have \(l_{i}^{*}\nk{y}=\frac{y}{\alpha_{i}}\nk{\ln\nk{y}-1}\) and \(\nk{l_{i}^{*}}^{\prime}\nk{y}=\frac{\ln\nk{y}}{\alpha_{i}}\). This yields
  \begin{equation}\label{eq:lambdastarexponentialcase}
    \lambda^{*}=\theta{B}
  \end{equation}
  and
  \begin{equation}\label{eq:penaltyexponentialcase}
    \alpha_{\Lambda,B}\nk{\Xi}=\sum_{i=1}^{n}\frac{1}{\alpha_{i}}\nk{H\nk{\bQ_{i}^{\bar{X}}\mid\bP}+\ln\nk{\theta{B}}},
  \end{equation}
  where \(\bQ_{i}^{\bar{X}}\) is the probability measure with density \(\Xiopt_{i}\).
  This certain structure of the penalty function yields another interesting property of the systemic risk measure \(\brisk^{ex}\). Since we already know that the optimal solution to the dual problem~\ref{eq:biaginisysdual} only distinguishes between the groups, we are able to decompose \(\brisk^{ex}\) into a sum of systemic risk measures. More precisely, if we set \(B_{j}=\frac{\theta}{\theta_{j}}B\) for \(j\in\gk{1,\ldots,h}\) and consider the systemic risk measures
  \[
    \brisk^{ex}_{j}\nk{\bar{Z}}=\inf\gk{\sum_{i\in{I_{j}}}{Y}_{i}\mid\bar{Y}\in\cC^{I_j},\bE\ek{\sum_{i\in{I_{j}}}\frac{1}{\alpha_{i}}\exp\nk{\alpha_{i}\nk{{Z}_i-{Y}_i}}}=B_{j}},
  \]
  for \(\bar{Z}\in\underset{i\in{I}_{j}}{\bigtimes}M^{\phi_{i}}\). Now, for \(\bar{X}\in\MPhi{}\), let \(\pi_{I_{j}}\circ\bar{X}=\nk{\bar{X}_{\tilde{n}_{j-1}},\ldots,\bar{X}_{\tilde{n}_{j}}}\). Then we have
  \[
    \brisk^{ex}\nk{\bar{X}}=\sum_{j=1}^{h}\brisk^{ex}_{j}\nk{\pi_{I_{j}}\circ\bar{X}}.
  \]
  The interpretation of this property is that the impact between the groups is fully determined through their risk aversion coefficients. Additionally, the system-wide threshold \(B\) effects every group. So, in other words, the potential outcomes of each group does not affect the risk of the other groups. This property could be seen as a clear indicator that this type of systemic risk measure is only useful if we consider systems consisting of only one group.
\end{ex}

\begin{rem}\label{rem:biaginiidentification}
  For the systemic risk measure \(\brisk^{ex}\), the set \(\cC{}\) satisfies all the assumptions of Theorem~\ref{theo:biaginisysnewproperties}, if the system consists of exactly one group. Since \(\brisk^{ex}\mid_{\bR^{n}}\nk{\nk{\Mexp{}}^{n}}=\Mexp{}\) we are able to apply Theorem~\ref{theo:sysriskcomposition}. On the other hand,
  Examples~\ref{ex:exponentialcase} and~\ref{ex:lambdastarandpenaltyexponentialcase} together yield, that the systemic risk measure \(\brisk^{ex}\) takes the form 
  \[
    \brisk^{ex}\nk{\bar{X}}=\frac{1}{\theta}\ln\bE\ek{\exp\nk{\theta\sum_{i=1}^{n}X_{i}}}-c=\nk{\rhosing^{entr}\circ\Lambda^{sum,c}}\nk{\bar{X}}=\rho^{ses,c}\nk{\bar{X}}
  \]
  for \(c=\frac{1}{\theta}\ln\nk{\theta{B}}\). 
  In the situation where more groups are relevant, we simply split the system into these groups and compute the systemic risk as presented in the previous example. In this case each systemic risk measure \(\brisk^{ex}_{j}\) for subgroup \(j\) uses the set \(\cC^{I_{j}}\) which again satisfies all the assumptions of Theorem~\ref{theo:biaginisysnewproperties}. Therefore, we are able to decompose all \(\brisk^{ex}_{j}\). The only connection between the groups is captured in the parameter \(\theta{}\).
\end{rem}


\section{CARs for Systemic Risk Measures}\label{sec:carsforsys}
The question of how to allocate the risk of the whole system to its constituent parts plays an important role. Let us start by giving a formal definition of CARs for systemic risk measures.

\begin{defi}\label{def:syscar}
 Given a systemic risk measure \(\rho{}\) on \(\productspace{}\), a systemic CAR is a map \(CS:\productspace\times\productspace\to\bR{}\) such that \(CS(\bar{X};\bar{X})=\rho(\bar{X})\) for every \(\bar{X}\in\productspace{}\). We say \(CS\) is a systemic CAR with respect to \(\rho{}\) and call the vector
 \[
  CS(\bar{X}):=\nk{CS(e_1{X}_{1},\bar{X}),\ldots,CS(e_n{X}_{n},\bar{X})}
 \]
 systemic capital allocation.
\end{defi}
\(CS(\bar{Y};\bar{X})\) describes the portion of risk which is carried by \(\bar{Y}\) considered as a sub-system of \(\bar{X}\).
The condition \(CS(\bar{X};\bar{X})=\rho(\bar{X})\) means that the capital allocated to the whole system \(\bar{X}\) is exactly the risk capital of \(\bar{X}\), i.e. \(\rho(\bar{X})\). We say a systemic CAR fulfills the full allocation property if the following holds:
\emph{
 \begin{enumerate}[label= (CS\arabic*)]
  \item\label{sysfullallocation} Full allocation: \(CS(\bar{X};\bar{X})=\sum_{i=1}^n CS(e_{i}{X}_{i};\bar{X})\), for every \(\bar{X}\in\productspace{}\).
 \end{enumerate}
}
The full allocation property~\ref{sysfullallocation} represents some kind of fairness condition. On the one hand, the regulator of the system can guarantee that the risk of the system is completely allocated to the participants. On the other hand, it is not over conservative in the sense that the allocated is no greater than the risk that actually occurs. Fairness properties on the level of the single institutions will be discussed later.

Both types of systemic risk measures studied in this paper admit a natural systemic CAR connected to optimal solutions of the corresponding dual problems.


\subsection{CARs for First Aggregate Systemic Risk Measures}\label{subseccarsforfirstaggregatesystemicriskmeasures}

If the dual problem~\eqref{eq:sysdualposhom} of a positively homogeneous systemic risk measure has an optimal solution \(\xixiopt{}\) there is a natural way to define a systemic CAR.\@ We simply set
\[
 CS^{\Xiopt}\nk{e_{i}{X}_{i};\bar{X}}=\npairing{e_{i}{X}_{i}}{\Xiopt}=\pairing{{X}_{i}}{\Xiopt_{i}}.
\]
Obviously, \(CS^{\Xiopt}\) is a systemic CAR and since
\[
 CS^{\Xiopt}\nk{\bar{X};\bar{X}}=\npairing{\bar{X}}{\Xiopt}=\sum_{i=1}^{n}\pairing{X_{i}}{\Xiopt_{i}}=\sum_{i=1}^{n}CS^{\Xiopt}\nk{e_{i}{X}_{i};\bar{X}}
\]
it also satisfies the full allocation property~\ref{sysfullallocation}. Since in general \(CS^{\Xiopt}(\bar{X})\not=CS^{\tilde{\Xi}^{\bar{X}}}(\bar{X})\) for two distinct optimal solutions \(\xixiopt{}{}\) and \(\nk{\tilde{\xi}^{\bar{X}},\tilde{\Xi}^{\bar{X}}}\), one has to decide which systemic capital allocation is optimal in some sense. However, if the dual problem~\eqref{eq:sysdualposhom} has a unique optimal solution \(\xixiopt{}\) the corresponding systemic CAR is unique and therefore the systemic capital allocation is unique. In addition, Corollary~\ref{cor:sysdirectionalforposhomunique} tells us that we are able to compute \(\Xiopt{}\) as the Gâteaux derivative of \(\rho{}\) at \(\bar{X}\). Therefore, this systemic CAR can be seen as a generalization of the Euler principle for single-firm risk measures (see~\cite{tasche2007capital} for more details on the Euler principle).
If the systemic risk measure is not positively homogeneous the situation slightly changes. But we are still able to define a systemic CAR based on optimal solutions to the dual problem~\eqref{eq:sysdualconvex}. In contrast to the positively homogeneous case one needs to apportion the additional penalty term. Therefore, the key task is to find fair rules.
Let \(\xixiopt{}{}\) be an optimal solution to the dual problem~\eqref{eq:sysdualconvex}. Then
\[
 CS^{\Xiopt}\nk{e_{i}{X}_{i};\bar{X}}=\npairing{e_{i}{X}_{i}}{\Xiopt}-\gamma_i\alpha\xixiopt=\pairing{X_{i}}{\Xiopt_{i}}-\gamma_i\alpha\xixiopt,
\]
where \(\gamma_i\), \(i\in\oneton{}\), are chosen such that \(\sum_{i=1}^{n}\gamma_i=1\), is a systemic CAR which fulfills the full allocation property~\ref{sysfullallocation}.~\cite{kromer2016systemic} presented some reasonable choices for \(\gamma{}\).

\begin{ex}\label{ex:penaltyallocationentropic}
  Consider the systemic risk measure \(\rho^{ses,c}\) presented in example~\ref{ex:sysentropic}. We have already seen that this systemic risk measure has a unique optimal solution to the dual problem given by 
  \[
    \xixiopt{}=\nk{\nabla\rhosing^{entr}\nk{\sum_{i=1}^{n}X_{i}},1_{n}\nabla\rhosing^{entr}\nk{\sum_{i=1}^{n}X_{i}}}.
  \]
  The penalty function is given by 
  \begin{align*}
    \alpha\xixiopt
    &=\frac{1}{\theta}H\nk{\bQ^{\bar{X}}\mid\bP}+c\\
    &=\frac{1}{\theta}\bE\ek{\xiopt\ln\nk{\xiopt}}+c\\
    &=\bE\ek{\xiopt\sum_{i=1}^{n}X_{i}}-\frac{1}{\theta}\ln\nk{\bE\ek{\exp\nk{\theta\sum_{i=1}^{n}X_{i}}}}+c,
  \end{align*}
  where \(\bQ^{\bar{X}}\) is the probability measure with density \(\xiopt{}\). 
  A natural way to apportion the penalty term is to use the individual risk aversion parameter, i.e.
  \[
    \gamma_{i}=\frac{\theta}{\alpha_{i}}.
  \]
  \(\gamma_{i}\) is simply the contribution of participant \(i\) to the systemic risk aversion parameter. Therefore, risk friendly participants are penalized more than risk averse participants. Now we have 
  \[
    CS^{\Xi}\nk{e_{i}X_{i};\bar{X}}=\bE\ek{\xiopt\nk{X_{i}-\frac{\theta}{\alpha_{i}}\sum_{i=1}^{n}X_{i}-\frac{1}{\alpha_{i}}\ln\nk{\bE\ek{\exp\nk{\theta\sum_{i=1}^{n}X_{i}}}}-\frac{\theta}{\alpha_{i}}c}}.
  \]
  Let us denote 
  \[
    Y_{i}^{\bar{X}}=X_{i}-\frac{\theta}{\alpha_{i}}\sum_{i=1}^{n}X_{i}-\frac{1}{\alpha_{i}}\ln\nk{\bE\ek{\exp\nk{\theta\sum_{i=1}^{n}X_{i}}}}-\frac{\theta}{\alpha_{i}}c.
  \]
  Then
  \[
    \sum_{i=1}^{n}Y_{i}^{\bar{X}}=\rho^{ses}\nk{\bar{X}}
  \]
  and the vector \(\bar{Y}^{\bar{X}}=\nk{Y_{1}^{\bar{X}},\ldots,Y_{n}^{\bar{X}}}\) can be seen as a scenario dependent allocation. In~\ref{ex:optimalYexponentialcase} we will see that \(\bar{Y}^{\bar{X}}\) is exactly the unique optimal solution to~\eqref{eq:exponentialcase} (h=1) if we choose \(c=\frac{1}{\theta}\ln\nk{\theta{}B}\). For \(h>1\) we can derive \(\bar{Y}^{\bar{X}}\) with the same procedure as already mentioned in~\ref{rem:biaginiidentification}.
\end{ex}


\subsection{CARs for First Inject Capital Systemic Risk Measures}\label{subsec:carsforfirstinjectcapitalsystemicriskmeasures}

The canonical way to find a CAR for systemic risk measures of type~\eqref{eq:sysbiaginiran} is provided by simply solving the optimization problem. If one can guarantee the existence of such an optimal allocation \(\bar{Y}_{\bar{X}}\) it is of explicit form as stated in the following proposition.
\begin{prop}
 Suppose that there exists an optimal allocation \(\bar{Y}^{\bar{X}}\in\cC{}\) to~\eqref{eq:biaginisysunderassumption}. Let \(\Xiopt{}\) be an optimal solution to the dual problem~\eqref{eq:biaginisysdual}. Then
  \begin{equation}\label{eq:biaginioptimalallocation}
  Y^{\bar{X}}_{i}=X_{i}-\nk{l_{i}^{*}}^{\prime}\nk{\lambda^{*}\Xiopt_{i}}
 \end{equation}
 and \(\Xiopt{}\) is the unique optimal solution to the dual problem~\eqref{eq:biaginisysdual}.
\end{prop}
\begin{proof}
 See~\cite{biagini2020fairness} Proposition 4.11 and Corollary 4.13.
\end{proof}
In order to guarantee the existence of an optimal allocation \(\bar{Y}^{\bar{X}}\in\cC_{0}\)
one has to make further assumptions on the set \(\cC_{0}\subseteq\cC\nk{\bR}\).
\begin{defi}\label{defi:truncation}
 A set \(\cC_{0}\subseteq\cC\nk{\bR}\) is closed under truncation if for each \(\bar{Y}\in\cC_{0}\) there exist \(m^{\bar{Y}}\in\bN{}\) and \(\bar{c}^{\bar{Y}}\in\bR^{n}\) such that \(\sum_{i=1}^{n}\bar{c}_{i}^{\bar{Y}}=\sum_{i=1}^{n}Y_{i}:=c^{\bar{Y}}\in\bR{}\) and for all \(m\geq{m}^{\bar{Y}}\)
 \begin{equation}\label{eq:truncation}
  \bar{Y}^{m}:=\bar{Y}\ind{\gk{\cap_{i=1}^{n}\gk{\abs{\bar{Y}^{n}}<m}}}+\bar{c}^{\bar{Y}}\ind{\gk{\cap_{i=1}^{n}\gk{\abs{\bar{Y}^{n}}\geq{m}}}}\in\cC_{0}.
 \end{equation}
\end{defi}
\begin{theo}\label{theo:optimaly}
 Suppose that \(\cC_{0}\) is closed under convergence in probability and closed under truncation. Then there exists a unique optimal solution to~\eqref{eq:biaginisysunderassumption} which coincides with
 \begin{equation}
  X_{i}-\nk{l_{i}^{*}}^{\prime}\nk{\lambda^{*}\Xiopt}\in\cC_{0}\cap\nk{\Lone\probspace{}\cap\Lone\nk{\Omega,\cF,\bQopt_{i}}},
 \end{equation}
 where \(\Xiopt{}\) is the unique optimal solution to the dual problem~\eqref{eq:biaginisysdual} and \(\bQopt_{i}\) is the probability measure with density \(\Xiopt_{i}\).
\end{theo}
\begin{proof}
 See~\cite{biagini2020fairness} Theorem 4.19 and Corollary 4.23.
\end{proof}

\begin{rem}\label{rem:systemiccapitalallocationbiagini}
  In the situation of Theorem~\ref{theo:optimaly}, a suitable systemic capital allocation is given by
  \begin{equation}\label{eq:systemiccapitalallocationbiagini}
    CS^{\nk{\bar{Y}^{\bar{X}},\Xiopt}}\nk{\bar{X}}=\nk{\pairing{Y_{1}^{\bar{X}}}{\Xiopt_{1}},\ldots,\pairing{Y_{n}^{\bar{X}}}{\Xiopt_{n}}}
  \end{equation}
  The vector of systemic probability measures \(\bar{\bQ}^{\bar{X}}\) corresponding to the vector of probability densities \(\Xiopt{}\) is the right one to give a fair valuation of the optimal random allocations (see~\cite{biagini2020fairness}). However, it is not clear which portion of risk is carried by an arbitrary subsystem \(\bar{Y}\) of \(\bar{X}\).
\end{rem}

\begin{ex}\label{ex:optimalYexponentialcase}
For the systemic risk measure presented in~\ref{ex:exponentialcase} we are able to compute the optimal allocation \(\bar{Y}^{\bar{X}}\). For \(j\in\gk{1,\ldots,h}\) and \(i\in{}I_{j}\) it takes the form
\begin{equation}\label{eq:optimalYexponentialcase}
  Y^{\bar{X}}_{i}=\bar{X}_{i}-\frac{\theta_{j}}{\alpha_{i}}\sum_{k\in{I}_{j}}X_{k}+\frac{1}{\alpha_{i}}\ln\nk{\bE\ek{\exp\nk{\theta_{j}\sum_{k\in{I}_{j}}X_{k}}}}-\frac{\ln\nk{\theta{B}}}{\alpha_{i}}.
\end{equation}\label{eq:allocationtogroupjexponentialcase}
The allocation to group \(j\) is given by
\begin{equation}
  d_{j}=\sum_{i\in{I}_{j}}\bar{Y}^{\bar{X}}_{i}=\frac{1}{\theta_{j}}\ln\nk{\bE\ek{\exp\nk{\theta_{j}\sum_{k\in{I}_{j}}X_{k}}}}-\frac{\ln\nk{\theta{B}}}{\theta_{j}}.
\end{equation}
\end{ex}


\subsection{Aumann-Shapley CAR for Systemic Risk Measures}\label{subsec:ascarforsystemicriskmeasures}

The CARs presented in the previous subsections are specified for the different types of systemic risk measures. A different way is to define systemic CARs in the spirit of the Aumann-Shapley CAR for single-firm risk measures. If the systemic risk measure \(\rho{}\) is Gâteaux differentiable at \(\gamma\bar{X}\) for all \(\gamma\in\zeroone{}\) we can define
\begin{equation}\label{eq:sysas}
 CS^{AS}\nk{\bar{Y},\bar{X}}=\int_{0}^{1}\delta\rho\nk{\gamma\bar{X},\bar{Y}}d\gamma.
\end{equation}
If we can assume that \(\rho\nk{0}=0\), we have
\begin{align*}
 \rho\nk{\bar{X}}
  & =\rho\nk{1\cdot\bar{X}}-\rho\nk{0\cdot\bar{X}}                                                                             \\
  & =\int_0^1\frac{d}{d\gamma}\nk{\rho\nk{\gamma\bar{X}}}d\gamma                                                               \\
  & =\int_0^1\lim_{\varepsilon\to0}\frac{\rho\nk{\nk{\gamma+\varepsilon}\bar{X}}-\rho\nk{\gamma\bar{X}}}{\varepsilon}d\gamma   \\
  & =\int_0^1\lim_{\varepsilon\to0}\frac{\rho\nk{\gamma\bar{X}+\varepsilon\bar{X}}-\rho\nk{\gamma\bar{X}}}{\varepsilon}d\gamma \\
  & =\int_0^1\delta\rho\nk{\gamma\bar{X},\bar{X}}d\gamma                                                                       \\
  & =CS^{AS}\nk{\bar{X},\bar{X}}.
\end{align*}
The major benefit of this systemic CAR is that we are able to apply it to both types of systemic risk measures. The full allocation property~\ref{sysfullallocation} is clearly satisfied if and only if the systemic risk measure is Gâteaux differentiable with derivative (i.e.\ the Gâteaux differential of \(\rho{}\) is a linear mapping in its second argument).

For systemic risk measures of type~\eqref{eq:syskromer} the chain rule for Gâteaux differentials applies. If we can additionally suppose that the corresponding single-firm risk measure \(\rhosing{}\) is Gâteaux differentiable at \(\Lambda\nk{\gamma\bar{X}}\) for all \(\gamma\in\zeroone{}\) with derivative \(\nabla\rhosing\nk{\Lambda\nk{\gamma\bar{X}}}\) and that the corresponding AR is Gâteaux differentiable at \(\gamma\bar{X}\) for all \(\gamma\in\zeroone{}\) we obtain
\begin{equation}\label{eq:sysaschainrule}
 CS^{AS}\nk{\bar{Y},\bar{X}}=\int_{0}^{1}\pairing{\nabla\rhosing\nk{\Lambda\nk{\gamma\bar{X}}}}{\delta\Lambda\nk{\gamma\bar{X},\bar{Y}}}d\gamma.
\end{equation}

The differentiability conditions on \(\rho{}\) and on \(\rhosing{}\) and \(\Lambda{}\) to derive~\eqref{eq:sysas} and~\eqref{eq:sysaschainrule} are already very restrictive. Moreover, the full allocation property~\ref{sysfullallocation} in~\eqref{eq:sysaschainrule} is satisfied if and only if \(\delta\Lambda{}\) is linear in its second argument. Since many reasonable choices for ARs do not satisfy this assumption,\@~\eqref{eq:sysaschainrule} is not a good choice for a systemic CAR in general.

\begin{ex}\label{ex:gateauxderivativearloss}
 Let us consider the AR \(\Lambda^{\text{loss}}\) presented in~\eqref{eq:arlossesonly}. Since \(\Lambda^{\text{loss}}\) is simply \(\sum_{i=1}^{n}F\nk{X_{i}}\), with \(F:\cX\to\cX{},X\mapsto\nk{X}^{+}\),
 we only need to control if the right directional derivative of \(F\) is linear in its second argument  in order to check whether \(\Lambda^{\text{loss}}\) is Gâteaux differentiable. We obtain
 \begin{align*}
  \delta_{+}F\nk{X,V}
   & =\lim_{h\to0^{+}}\frac{F\nk{X+hV}-F\nk{X}}{h}       \\
   & =\lim_{h\to0^{+}}\frac{\nk{X+hV}^{+}-\nk{X}^{+}}{h} \\
   & =\begin{cases}
   V     & \text{on }\gk{X>0}  \\
   V^{+} & \text{on }\gk{X=0}  \\
   0     & \text{on }\gk{X<0}.
  \end{cases}
 \end{align*}
 Obviously, on \(\gk{X=0}\) the linearity is lost and therefore \(\Lambda^{\text{loss}}\) is not Gâteaux differentiable.
\end{ex}

An alternative way for systemic risk measures of type~\eqref{eq:syskromer} to define a systemic CAR in the spirit of Aumann-Shapley is given by
\begin{equation}\label{eq:sysasalternative}
 \bar{CS}^{AS}\nk{\bar{Y},\bar{X}}=\int_{0}^{1}\delta\rhosing\nk{\gamma\Lambda\nk{\bar{X}},\Lambda\nk{\bar{Y}}}d\gamma.
\end{equation}
In this situation we only need the single-firm risk measure \(\rhosing{}\) to be Gâteaux-differentiable at \(\gamma\Lambda\nk{\bar{X}}\) for all \(\gamma\in\zeroone{}\). Again, if \(\rhosing{}\) has a derivative \(\nabla\rhosing\nk{\gamma\Lambda\nk{\bar{X}}}\) we can write~\eqref{eq:sysasalternative} as
\begin{equation}\label{eq:sysasalternativederivative}
 \bar{CS}^{AS}\nk{\bar{Y},\bar{X}}=\int_{0}^{1}\pairing{\nabla\rhosing\nk{\gamma\Lambda\nk{\bar{X}}}}{\Lambda\nk{\bar{Y}}}d\gamma.
\end{equation}
The Full allocation property~\ref{sysfullallocation}, clearly, is only satisfied if the AR \(\Lambda{}\) is additive at \(\bar{X}\).

Let us continue by presenting some examples.

\begin{ex}\label{ex:aumannshapleykromerandbiagini}
  Consider the systemic risk measure \(\rho^{\text{ses}}\) defined in~\eqref{eq:sysentropicsum}. Since both, \(\rhosing^{entr}\) and \(\Lambda^{\text{sum}}\) are Gâteaux differentiable with derivative and \(\rho^{\text{ses}}\nk{0}=0\) we are able to apply~\eqref{eq:sysas}. It yields
  \begin{equation*}
    CS^{AS}_{ses}\nk{\bar{Y},\bar{X}}=\int_{0}^{1}\bE\ek{\sum_{i=1}^{n}Y_{i}\frac{\exp\nk{\theta\gamma\sum_{i=1}^{n}X_{i}}}{\bE\ek{\exp\nk{\theta\gamma\sum_{i=1}^{n}X_{i}}}}}d\gamma.
  \end{equation*}
  Note, that in this situation the full allocation property~\ref{sysfullallocation} is satisfied.

  Now consider the systemic risk measure \(\brisk^{ex}\) presented in Example~\ref{ex:exponentialcase} for one group (h=1). As seen in~\ref{ex:exponentialcase} \(\brisk^{ex}\) is Gâteaux differentiable with derivative and,
 since \(\brisk^{ex}\nk{0}=0\), we are able to apply~\eqref{eq:sysas}. We obtain
 \begin{equation*}
  CS^{AS}_{ex}\nk{\bar{Y},\bar{X}}=\int_{0}^{1}\bE\ek{\sum_{i=1}^{n}Y_{i}\frac{\exp\nk{\theta\gamma\sum_{i=1}^{n}X_{i}}}{\bE\ek{\exp\nk{\theta\gamma\sum_{i=1}^{n}X_{i}}}}}d\gamma.
 \end{equation*}
 Again, the full allocation property~\ref{sysfullallocation} is satisfied. Obviously, 
 \[
  CS^{AS}_{ses}\nk{\bar{Y},\bar{X}}=CS^{AS}_{ex}\nk{\bar{Y},\bar{X}}
 \]
 for all \(\bar{Y}\) and therefore
 \[
  \rho^{\text{ses}}\nk{\bar{X}}=CS^{AS}_{ses}\nk{\bar{X},\bar{X}}=CS^{AS}_{ex}\nk{\bar{X},\bar{X}}=\brisk^{ex}\nk{\bar{X}}.
 \]
 So if the full allocation property~\ref{sysfullallocation} is satisfied, the Aumann-Shapley CAR yields an alternative approach to identify systemic risk measures of type~\ref{eq:sysbiaginiran} with systemic risk measures of type~\ref{eq:syskromer}.
\end{ex}

\begin{ex}
 Let us consider the systemic risk measure \(\rho^{\text{sel}}\) defined in~\eqref{eq:sysentropicloss}. As shown in Example~\ref{ex:gateauxderivativearloss} the AR \(\Lambda^{\text{loss}}\nk{\bar{X}}\) is not Gâteaux differentiable at 0. Therefore, a computation of a systemic CAR via~\eqref{eq:sysaschainrule} fails. But we are still able to compute a systemic CAR via~\eqref{eq:sysasalternativederivative} and, we obtain
 \begin{align*}
  \bar{CS}^{AS}\nk{\bar{Y},\bar{X}}=\int_0^1\bE\ek{\Lambda^{\text{loss}}\nk{\bar{Y}}\frac{\exp\nk{\theta\gamma\Lambda^{\text{loss}}\nk{\bar{X}}}}{\bE\ek{\exp\nk{\theta\gamma\Lambda^{\text{loss}}\nk{\bar{X}}}}}}d\gamma.
 \end{align*}
\end{ex}

\begin{ex}
 Consider the systemic risk measure defined in Example~\ref{ex:biaginiinkromersetting}. Since both, \(\Lambda^{exut}\) and \(\rhosing{}\), are Gâteaux differentiable we are able to compute the systemic capital allocation via~\eqref{eq:sysaschainrule} if we can guarantee that \(\rho\nk{0}=0\). This condition is clearly satisfied if we set \(B=\sum_{i=1}^{n}\frac{1}{\alpha_{i}}\) and in this case we obtain
 \begin{align*}
  CS^{AS}\nk{\bar{Y},\bar{X}}
  &=\sum_{i=1}^{n}\int_{0}^{1}\bE\ek{Y_i\exp\nk{\gamma\alpha_{i}X_i}}d\gamma.\\
  &=\sum_{i=1}^{n}\bE\ek{\frac{Y_i}{\alpha_{i}X_{i}}\nk{\exp\nk{\gamma\alpha_{i}X_i}-1}}
 \end{align*}
Obviously, the full allocation property~\ref{sysfullallocation} is satisfied. The risk allocation to institution \(i\in\gk{\oneton}\) is given by
\begin{align*}
 CS^{AS}\nk{e_{i}X_{i},\bar{X}}
 &=\bE\ek{\frac{1}{\alpha_{i}}\nk{\exp\nk{\gamma\alpha_{i}X_{i}}-1}}\\
 &=\rho\nk{e_{i}X_{i}}+\sum_{j\not=i}\frac{1}{\alpha_{j}}.
\end{align*}
 An alternative way to derive a systemic capital allocation for this specific systemic risk measure is given by~\eqref{eq:sysasalternativederivative}. Again, we set \(B=\sum_{i=1}^{n}\frac{1}{\alpha_{i}}\) and obtain
 \begin{align*}
  \bar{CS}^{AS}\nk{\bar{Y},\bar{X}}
  &=\int_{0}^{1}\bE\ek{\sum_{i=1}^{n}\frac{1}{\alpha_{i}}\exp\nk{\alpha_{i}Y_{i}}}d\gamma{}\\
  &=\bE\ek{\sum_{i=1}^{n}\frac{1}{\alpha_{i}}\exp\nk{\alpha_{i}Y_{i}}}\\
  &=\rho\nk{\bar{Y}}+B.
 \end{align*}
 This systemic capital allocation does not fulfill the full allocation property~\ref{sysfullallocation}.
\end{ex}


\appendix

\section{Appendix}

\subsection{Orlicz Spaces}\label{appendix:orliczspaces}

A special class of function spaces which include the \(\Lp\probspace{}\)- spaces are the Orlicz spaces. In a risk measurement framework these spaces appear to be the right choice if the risk measure is connected with a loss function (see Example~\ref{ex:sysentropic}). For this purpose, let \(l\colon\bR\to\bR{}\) be a convex and increasing function with \(\lim_{x\to\infty}\frac{l\nk{x}}{x}=\infty{}\).
Then the function \(\phi\colon\bR\to\interval[open right]{0}{\infty}\), given by \(\phi\nk{x}:=l\nk{\abs{x}}-l\nk{0}\) is finite valued, even and convex on \(\bR{}\) with \(\phi\nk{0}=0\) and \(\lim_{x\to\infty}\frac{\phi\nk{x}}{x}=\infty{}\), i.e., \(\phi{}\) is a strict Young function.
The Orlicz space \(\Lphi{}\) and the Orlicz heart \(\Mphi{}\) are defined by
\begin{align*}
 \Lphi & \colon=\gk{X\in\Lzero\mid\bE\ek{\phi\nk{\alpha{X}}}<\infty\text{ for some }\alpha>0}, \\
 \Mphi & \colon=\gk{X\in\Lzero\mid\bE\ek{\phi\nk{\alpha{X}}}<\infty\text{ for all }\alpha>0}.
\end{align*}
If these spaces are endowed with the Luxemburg norm, they are Banach lattices and thus fit in our general framework (see~\cite{rubshtein2016foundations} Theorem 13.2.2). The (topological) dual space of \(\Mphi{}\) is the Orlicz space \(L^{\phi^{*}}\), where \(\phi^{*}\) is the convex conjugate of \(\phi{}\) defined by
\begin{equation*}
 \phi^{*}\nk{y}:=\sup_{x\in\bR}\gk{xy-\phi\nk{x}},\quad{y}\in\bR,
\end{equation*}
and \(\phi^{*}\) itself is also a strict Young function.
Remark, that \(\Linf{}\subseteq\Mphi\subseteq\Lphi\subseteq\Lone{}\). Additionally, Fenchels inequality yields for \(\frac{d\bQ}{d\bP}\in{L}^{\phi^{*}}\), that \(\Lphi\subseteq\Lone\nk{\Omega,\cF,\bQ}\). Now for two Young functions \(\phi_{1},\phi_{2}\) we say that \(\phi_{1}\) majorizes \(\phi_{2}\) (\(\phi_{1}\succ\phi{2}\)) if 
\[
  \phi_{2}\nk{x}\leq{b}\phi_{1}\nk{ax}\qquad\forall{x}\geq0
\]
for some \(b>0\) and \(a>0\). In this case we have \(M^{\phi_{1}}\subseteq{M}^{\phi_{2}}\) (see~\cite{rubshtein2016foundations} Theorem 16.2.1.). For \(l_{1},\ldots,l_{n}\colon\bR\to\bR{}\) with associated Young functions \(\phi_{1},\ldots,\phi_{n}\), we simply write
\begin{equation*}
 \MPhi:=M^{\phi_{1}}\times\cdots\times{M}^{\phi_{n}}\quad\text{ and }\quad \LPhi:=L^{\phi_{1}}\times\cdots\times{L}^{\phi_{n}}.
\end{equation*}

\subsection{Gâteaux Differentials}\label{appendix:gateauxdifferentials}
We will need a generalization of the concept of directional derivatives. In the following \(\cU,\cV{}\) and \(\cW{}\) are arbitrary locally convex topological vector spaces.

\begin{defi}
 Let \(F:\cU\to\cV{}\), \(\cO\subset\cU{}\), \(X\in\text{int}\nk{\cO}\) and \(U\in\cU{}\). Define
 \[
  \delta_{+}F\nk{X,U}=\lim_{h\to0^{+}}\frac{F\nk{X+hU}-F\nk{X}}{h}.
 \]
 If the limit exists we call \(\delta_{+}F\nk{X,U}\) the right directional derivative of \(F\) at \(X\) in direction \(U\).
\end{defi}
\begin{rem}
 In the same setup, we are able to define
 \[
  \delta_{-}F\nk{X,U}=\lim_{h\to0^{-}}\frac{F\nk{X+hU}-F\nk{X}}{h}.
 \]
 If the limit exists we call \(\delta_{-}F\nk{X,U}\) the left directional derivative of \(F\) at \(X\) in direction \(U\). However, the left directional derivative can be expressed in terms of the right directional derivative, i.e.
 \[
  \delta_{-}F\nk{X,U}=-\delta_{+}F\nk{X,-U}.
 \]
 Therefore, we call \(\delta_{+}F\nk{X,U}\) the directional derivative and omit the term right.
\end{rem}
If both, \(\delta_{+}F\nk{X,U}\) and \(\delta_{-}F\nk{X,U}\), exist and coincide for all \(U\in\cU{}\) we call the mapping \(F\) \textit{Gâteaux differentiable}. The following definition recaptures this idea.

\begin{defi}\label{def:gateauxdifferential}
  Let \(F:\cU\to\cV{}\), \(\cO\subset\cU{}\), \(X\in\text{int}\nk{\cO}\) and \(U\in\cU{}\). Define
 \[
  \delta{F}\nk{X,U}=\lim_{h\to0}\frac{F\nk{X+hU}-F\nk{X}}{h}.
 \]
 If the limit exists for all \(U\in\cU{}\) we call the mapping \(\delta{F}\nk{X,\cdot}:\cU\to\cV{}\) Gâteaux differential\footnote{Sometimes it is called \emph{weak differential}.} of \(F\) at \(X\) and say that \(F\) is Gâteaux differentiable at \(X\). If \(F:\cU\to\bR{}\) is Gâteaux differentiable at \(X\) and there exists a \(\xi\in\cU^{\prime}\) such that
 \(\delta{F}\nk{X,\cdot}=\pairing{\cdot}{\xi}\) we call \(F\) Gâteaux differentiable with (Gâteaux) derivative \(\nabla{F}\nk{X}:=\xi{}\) at \(X\).
\end{defi}

The following Proposition collects some rules for Gâteaux differentials.
\begin{prop}\label{prop:rulesforgateaux}
 Let \(F,G:\cU\to\cV{}\), \(H:\cW\to\cU{}\), \(U\in\cU{}\) and \(W\in\cW{}\).
 \begin{compactenum}[(i)]
  \item Suppose that \(F\) and \(G\) are Gâteaux differentiable at \(X\in\cU{}\). Then
  \[
   \delta\nk{F\pm{G}}\nk{X,\cdot}=\delta{F}\nk{X,\cdot}\pm\delta{G}\nk{X,\cdot}.
  \]
  \item Suppose that \(F\) and \(G\) are Gâteaux differentiable at \(X\in\cU{}\). Then
  \[
   \delta\nk{FG}\nk{X,\cdot}=\delta{F}\nk{X,\cdot}G\nk{X}+F\nk{X}\delta{G}\nk{X,\cdot},
  \]
  where \(FG\) describes the element wise product of \(F\) and \(G\).
  \item Suppose that \(H\) is Gâteaux differentiable at \(Y\in\cW{}\) and \(G\) is Gâteaux differentiable at \(H\nk{Y}\). Then
  \[
   \delta\nk{G\circ{H}}\nk{Y,\cdot}=\delta{G}\nk{H\nk{Y},\delta{H}\nk{Y,\cdot}}
  \]
 \end{compactenum}
\end{prop}

\begin{proof}
 Part (i) follows directly from Definition~\ref{def:gateauxdifferential}. To prove part (ii) and (iii), first notice that 
 \[
   F\nk{X+hU}=F\nk{X}+h\delta{F}\nk{X,U}+o\nk{h},
 \]
 where \(o\nk{h}\) describes some \(q\) with 
 \[
   \lim_{h\to0}\frac{q}{h}=0.
 \]
 Now, for all \(U\in\cU{}\) we obtain 
 \begin{align*}
  \delta\nk{FG}\nk{X,U}
  &=\lim_{h\to0}\frac{F\nk{X+hU}G\nk{X+hU}-F\nk{X}G\nk{X}}{h}\\
  &=\lim_{h\to0}\frac{\nk{F\nk{X}+h\delta{F}\nk{X,U}+o\nk{h}}\nk{G\nk{X}+h\delta{G}\nk{X,U}+o\nk{h}}-F\nk{X}G\nk{X}}{h}\\
  &=\delta{F}\nk{X,U}G\nk{X}+F\nk{X}\delta{G}\nk{X,U}\\
  &\ \ \ \ +\lim_{h\to0}\frac{F\nk{X}o\nk{h}+G\nk{X}o\nk{h}+o\nk{h}^{2}}{h}\\
  &\ \ \ \ +\lim_{h\to0}\nk{h\delta{F}\nk{X,U}\delta{G}\nk{X,U}+\delta{F}\nk{X,U}o\nk{h}+\delta{G}\nk{X,U}o\nk{h}}\\
  &=\delta{F}\nk{X,U}G\nk{X}+F\nk{X}\delta{G}\nk{X,U},
 \end{align*}
 which proves part (ii). For part (iii) we have for all \(W\in\cW{}\) 
 \begin{align*}
  \delta\nk{G\circ{H}}\nk{Y,W}
  &=\lim_{h\to0}\frac{G\nk{H\nk{Y+hW}}-G\nk{H\nk{Y}}}{h}\\
  &=\lim_{h\to0}\frac{G\nk{H\nk{Y}+h\delta{H}\nk{Y,W}+o\nk{h}}-G\nk{H\nk{Y}}}{h}\\
  &=\lim_{h\to0}\frac{G\nk{H\nk{Y}+h\nk{\delta{H}\nk{Y,W}+h^{-1}o\nk{h}}}-G\nk{H\nk{Y}}}{h}\\
  &=\lim_{h\to0}\frac{G\nk{H\nk{Y}}+h\delta{G}\nk{H\nk{Y},\delta{H}\nk{Y,W}+h^{-1}o\nk{h}}+o\nk{h}-G\nk{H\nk{Y}}}{h}\\
  &=\delta{G}\nk{H\nk{Y},\delta{H}\nk{Y,W}}.
 \end{align*}
\end{proof}

\subsection{Subgradients}\label{appendix:subgradients}
In cases where we can not assume Gâteaux differentiability we have to work with the more general notion of \textit{subgradients}. By \(L\nk{\cU,\cV}\) we denote the real vector space of linear mappings \(T:\cU\to\cV{}\).
\begin{defi}
 Let \(F:\cU\to\cV{}\). The strong vector subdifferential of at \(X\in\cU{}\) is given by
 \begin{equation}\label{strongvectorsubdiff}
  \partial^{s}F\nk{X}:=\gk{T\in{L}\nk{\cU,\cV}\mid{T}\nk{U-X}\leq{F}\nk{U}-F\nk{X},\ \forall{U}\in\cU{}}.
 \end{equation}
\end{defi}

\begin{rem}
 For convex functionals \(F:\cU\to\bR\cup\gk{\infty}\) (which are not necessarily Gâteaux differentiable) the subdifferential at \(X\in\text{dom}\nk{F}=\gk{X\in\cU\mid F\nk{X}<+\infty}\) reduces to
 \begin{align*}
  \partial F\nk{X} 
  & =\gk{\xi\in\cU^{\prime}\mid\pairing{U-X}{\xi}\leq{F}\nk{U}-F\nk{X}\quad\forall{U}\in\cU} \\
  & =\gk{\xi\in\cU^{\prime}\mid\pairing{U}{\xi}+F\nk{X}\leq{F}\nk{X+U}\quad\forall{U}\in\cU} \\
 \end{align*}

 We say that \(F\) is sub-differentiable at \(X\in\text{dom}\nk{F}\) if \(\partial{F}\nk{X}\) is nonempty and call the elements of \(\partial{F}\nk{X}\) subgradients (at \(X\)). In~\cite{zalinescu2002convex} it is shown, that \(\partial{F}\nk{X}\) is a singleton if and only if \(F\) is Gâteaux differentiable at \(X\).
\end{rem}

\bibliography{lit}
\end{document}